\documentclass[a4paper,11pt]{article}
\usepackage{graphicx}
\usepackage{psfrag}
\usepackage{epsfig}
\usepackage{subfigure}

\usepackage{amssymb}
\usepackage{psfrag}
\usepackage{amsmath,amsthm,amssymb}

\usepackage{authblk}

\usepackage{hyperref}

\hypersetup{
    colorlinks,%
    citecolor=black,%
    filecolor=black,%
    linkcolor=black,%
    urlcolor=black
}

\numberwithin{equation}{section}

\newtheorem{dfn}{Definition}[section]
\newtheorem{rmk}{Remark}[section]
\newtheorem{thm}{Theorem}[section]
\newtheorem{prop}[thm]{Proposition}

\begin{document}

\title{{Discrete Exterior Geometry Approach\\to Structure-Preserving Discretization of\\Distributed-Parameter Port-Hamiltonian Systems}}

 \author[*]{Marko Seslija}
 \author[$\dag$]{Arjan van der Schaft}
 \author[*]{Jacquelien M.A. Scherpen}

 \affil[*]{Department of Discrete Technology and Production Automation, Faculty of Mathematics and Natural Sciences, University of Groningen, Nijenborgh 4, 9747 AG Groningen, The Netherlands, e-mail:~\{M.Seslija,\;J.M.A.Scherpen\}@rug.nl}
 \affil[$\dag$]{Johann Bernoulli Institute for Mathematics and Computer Science, University of Groningen, Nijenborgh 9, 9747 AG Groningen, The Netherlands, e-mail:~A.J.van.der.Schaft@rug.nl}

\maketitle

\begin{abstract}
This paper addresses the issue of structure-preserving discretization of open distri-buted-parameter systems with Hamiltonian dynamics. Employing the formalism of discrete exterior calculus, we introduce a simplicial Dirac structure as a discrete analogue of the Stokes-Dirac structure and demonstrate that it provides a natural framework for deriving finite-dimensional port-Hamiltonian systems that emulate their infinite-dimensional counterparts. The spatial domain, in the continuous theory represented by a finite-dimensional smooth manifold with boundary, is replaced by a homological manifold-like simplicial complex and its augmented circumcentric dual. The smooth differential forms, in discrete setting, are mirrored by cochains on the primal and dual complexes, while the discrete exterior derivative is defined to be the coboundary operator. This approach of discrete differential geometry, rather than discretizing the partial differential equations, allows to first discretize the underlying Stokes-Dirac structure and then to impose the corresponding finite-dimensional port-Hamiltonian dynamics. In this manner, a number of important intrinsically topological and geometrical properties of the system are preserved.
\end{abstract}

\section{Introduction}
The purpose of this paper is to propose a sound geometric framework for structure-preserving discretization of distributed-parameter port-Hamiltonian systems. Our approach to time-continuous spatially-discrete port-Hamiltonian theory is based on discrete exterior geometry and as such proceeds \emph{ab initio} by mirroring the continuous setting. The theory is not merely tied to the goal of discretization but rather aims to offer a sound and consistent framework for defining port-Hamiltonian dynamics on a discrete manifold which is usually, but not necessarily, obtained by discretization of a smooth Riemannian manifold.

The underlying structure of open distributed-parameter dynamical systems considered in this paper is a Stokes-Dirac structure \cite{vdSM02} and as such is being defined on a certain space of differential forms on a smooth finite-dimensional orientable, usually Riemannian, manifold with a boundary. The Stokes-Dirac structure generalizes the framework of the Poisson and symplectic structures by providing a theoretical account that permits the inclusion of varying boundary variables in the boundary problem for partial differential equations. From an interconnection and control viewpoint, such a treatment of boundary conditions is essential for the incorporation of energy exchange through the boundary, since in many applications the interconnection with the environment takes place precisely through the boundary. The same arguments apply to the finite-dimensional approximations of complex distributed-parameter systems. For numerical integration, simulation and control synthesis, it is of paramount interest to have finite approximations that can be interconnected to one another or via the boundary coupled to other systems, be they finite- or infinite-dimensional.

Most of the numerical algorithms for spatial discretization of distributed-parameter systems, primarily finite difference and finite element methods, fail to capture the intrinsic system structures and properties, such as symplecticity, conservation of momenta and energy, as well as differential gauge symmetry. Furthermore, some important results, including the Stokes theorem, fail to apply numerically and thus lead to spurious results. This loss of fidelity to preserve some inherent topological and geometric structures of the continuous models motivates a more geometry based approach.

The discrete approach to geometry goes back to Whitney, who in \cite{Whitney} introduced an isomorphism between simplicial and de Rham cohomology. More recent antecedents can be found, for instance, in \cite{Sen}, and also in the computational electromagnetism literature \cite{Bosavit,Bossavit1,Gross}. For a comprehensive historical summary we refer to the thesis \cite{Hirani} and references therein. The literature, however, seems mostly focused on discretization of systems with infinite spatial domains, boundaryless manifolds, and systems with zero boundary conditions. In this paper, we augment the definition of the dual cell complex in order to allow nonzero energy flow through boundary. 

A notable previous attempt to resolve the problem of structure-preserving discretization of port-Hamiltonian systems is \cite{Golo}, where the authors employ the mixed finite element method. Their treatment is restricted to the one-dimensional telegraph equation and the two-dimensional wave equation. Although it is hinted that the same methodology applies in higher dimensions and to the other distributed-parameter systems, the results are not clear. It is worth noting that the choice of the basis functions can have dramatic consequences on the numerical performance of the mixed finite element method; as the mesh is being refined, it easily may lead to an ill-conditioned finite-dimensional linear system \cite{Arnold}. The other undertaking on discretization of port-Hamiltonian systems can be found in \cite{vdSM08,vdSM09}, but the treatment is purely topological and is more akin to the graph-theoretical formulation of conservation laws. Furthermore, the authors in \cite{vdSM08,vdSM09} do not introduce a discrete analogue of the Stokes theorem and the entire approach is tied to the goal of preserving passivity.

Our approach is that of discrete exterior calculus \cite{Desbrun1,Desbrun,Hirani,Stern}, which has previously been applied to variational problems naturally arising in mechanics and electromagnetism. These problems stem from a Lagrangian, rather than Hamiltonian, modeling perspective and as such they conform to a multisymplectic structure \cite{Gotay,Marsden,Marsden1,Joris}, rather than the Stokes-Dirac structure. A crucial ingredient for the numeric integration is the asynchronous variational integrator for spatio-temporally discretized problems, whereas our approach spatially discretizes the Stokes-Dirac structure and allows imposing time-continuous spatially discrete dynamics. This apparent discrepancy between multisymplectic and the Stokes-Dirac structure-preserving discretization could be elevated by, for instance, defining Stokes-Dirac structure on a pseudo-Riemannian manifold to insure a treatment of space and time on equal footing, whilst keeping nonzero exchange through the boundary.

\vspace{0.15cm}
\noindent{\textbf{Contribution and outline of the paper.}} We begin by recalling the definition of the Stokes-Dirac structure and port-Hamiltonian systems. In order to make this paper as self-contained as possible for a variety of readers, we present a brief overview of the elementary discrete exterior geometry needed to define a discretized Stokes-Dirac structure and impose appropriate port-Hamiltonian dynamics. The third section is a brief summary of the essential definitions and results in discrete exterior calculus as developed in \cite{Desbrun1,Desbrun, Hirani}. The contribution of this paper in this regard is a proper treatment of the boundary of the dual cell complex. Namely, in order to allow the inclusion of nonzero boundary conditions on the dual cell complex, we offer a definition of the dual boundary operator that differs from the standard one. Such a construction leads to a discrete analogue of the integration by parts formula, which is a crucial ingredient in establishing a discrete Stokes-Dirac structure on a primal simplicial complex and its circumcentric dual. The main result is presented in Section~\ref{Sec4}, where we introduce the notion of simplicial Dirac structures on a primal-dual cell complex, and in the following section define port-Hamiltonian systems with respect to these structures. In Section~\ref{Sec:matrixrep} we give a matrix representation for the simplicial Dirac structures and linear port-Hamiltonian systems, for which we also establish bounds for the energy of discretization errors. Finally, we demonstrate how the simplicial Dirac structures relate to some spatially discretized distributed-parameter systems with boundary variables: Maxwell's equations on a bounded domain, a two-dimensional wave equation, and the telegraph equations. While the focus of this paper is not implementation of discrete exterior calculus in discretization of port-Hamiltonian systems, we have, nonetheless, taken some preliminary numerical investigations. We demonstrate the application of the developed machinery using the example of the telegraph equations.

Some preliminary results of this paper have been reported in \cite{SeslijaCDC}.

\section{Dirac Structures and Port-Hamiltonian Dynamics}\label{Sec2}
Dirac structures were originally developed in \cite{CourantWeinstein,Courant,Dorfman} as a generalization of symplectic and Poisson structures. The formalism of Dirac structure was employed as the geometric notion underpinning generalized power-conserving interconnections and thus allowing the Hamiltonian formulation of interconnected and constrained dynamical systems.

A constant Dirac structure can be defined as follows. Let $\mathcal{F}$, $\mathcal{E}$, and $L$ be linear spaces. Given a $f\in\mathcal{F}$ and an $e\in\mathcal{E}$, the pairing will be denoted by $\langle e | f\rangle\in L$. By symmetrizing the pairing, we obtain a symmetric bilinear form $\langle\!\langle,\rangle\!\rangle: \mathcal{F}\times \mathcal{E} \rightarrow L$ defined by
\begin{equation*}
\langle\!\langle (f_1,e_1) ,(f_2,e_2)\rangle\!\rangle=\langle e_1 |f_2\rangle+\langle e_2 | f_1\rangle\,.
\end{equation*}

\begin{dfn}
A Dirac structure is a linear subspace $\mathcal{D}\subset \mathcal{F} \times \mathcal{E}$ such that $\mathcal{D}=\mathcal{D}^{\perp}$, with $\perp $ standing for the orthogonal complement with respect to the bilinear form $\langle\!\langle,\rangle\!\rangle$.
\end{dfn}

It immediately follows that for any $(f,e)\in \mathcal{D}$
\begin{equation*}
0=\langle\!\langle (f,e) ,(f,e)\rangle\!\rangle=2 \langle e |f \rangle\,.
\end{equation*}
Interpreting $(f,e)$ as a pair of power variables, the condition $(f,e)\in \mathcal{D}$ implies power-conservation $\langle e| f\rangle=0$, and as such is \emph{terminus a quo} for the geometric formulation of port-Hamiltonian systems.

Much is known about finite-dimensional Dirac structures and their role in physics; however, hitherto there is no complete theory of Dirac structures for field theories. An initial contribution in this direction is made in the paper \cite{vdSM02}, where the authors introduce a notion of the Stokes-Dirac structure. This infinite-dimensional Dirac structure lays down the foundation for port-Hamiltonian formulation of a class of distributed-parameter systems with boundary energy flow. In this section, we provide a very brief overview of the Stokes-Dirac structure \cite{vdSM02}.

Throughout this paper, let $M$ be an oriented $n$-dimensional smooth manifold with a smooth $(n-1)$-dimensional boundary $\partial M$ endowed with the induced orientation, representing the space of spatial variables. By $\Omega^k(M)$, $k=0,1,\ldots,n$, denote the space of exterior $k$-forms on $M$, and by $\Omega^k(\partial M)$, $k=0,1,\ldots,n-1$, the space of $k$-forms on $\partial M$. A natural non-degenerative pairing between $\alpha\in\Omega^k(M)$ and $\beta\in\Omega^{n-k}(M)$ is given by $\langle \beta|\alpha\rangle=\int_M \beta \wedge \alpha$. Likewise, the pairing on the boundary $\partial M$ between $\alpha\in \Omega^k(\partial M)$ and $\beta\in \Omega^{n-k-1}(\partial M)$ is given by $\langle \beta|\alpha\rangle=\int_{\partial M} \beta \wedge \alpha$.

For any pair $p,q$ of positive integers satisfying $p+q=n+1$, define the flow and effort linear spaces by
\begin{equation*}
\begin{split}
\mathcal{F}_{p,q}=\,&\Omega^p(M)\times \Omega^q(M)\times \Omega^{n-p}(\partial M)\,\\
\mathcal{E}_{p,q}=\,&\Omega^{n-p}(M)\times \Omega^{n-q}(M)\times \Omega^{n-q}(\partial M)\,.
\end{split}
\end{equation*}
The bilinear form on the product space $\mathcal{F}_{p,q}\times\mathcal{E}_{p,q}$ is given by
\begin{equation}\label{eq-aj9}
\begin{split}
\langle\!\langle (\underbrace{f_p^1,f_q^1,f_b^1}_{\in \mathcal{F}_{p,q}},&\underbrace{e_p^1,e_q^1,e_b^1}_{\in \mathcal{E}_{p,q}}), (f_p^2,f_q^2,f_b^2,e_p^2,e_q^2,e_b^2)\rangle\!\rangle=\\
& \int_M \left( e_p^1\wedge f_p^2+e_q^1\wedge f_q^2+ e_p^2\wedge f_p^1+e_q^2\wedge f_q^1 \right)+\int_{\partial M} \left(e_b^1\wedge f_b^2+ e_b^2\wedge f_b^1\right)
  \,.
  \end{split}
\end{equation}

\begin{thm}
Given linear spaces $\mathcal{F}_{p,q}$ and $\mathcal{E}_{p,q}$, and bilinear form $\langle\!\langle, \rangle\!\rangle$, define the following linear subspace $\mathcal{D}$ of $\mathcal{F}_{p,q}\times \mathcal{E}_{p,q}$
\begin{equation}\label{eq-7Cont}
\begin{split}
\mathcal{D}=\big\{& (f_p,f_q,f_b,e_p,e_q,e_b)\in \mathcal{F}_{p,q}\times \mathcal{E}_{p,q}\big |\\
& \left(\begin{array}{c}f_p \\f_q\end{array}\right)=\left(\begin{array}{cc}0 & (-1)^{pq+1} {\mathrm{d}}\\ {\mathrm{d}} & 0\end{array}\right)\left(\begin{array}{c}e_p \\e_q\end{array}\right)\,,\\
& \left(\begin{array}{c}f_b \\e_b\end{array}\right)=\left(\begin{array}{cc}1&0\\0 & -(-1)^{n-q}\end{array}\right)\left(\begin{array}{c}e_p|_{\partial M} \\e_q|_{\partial M}\end{array}\right)
\big \}\,,
 \end{split}
\end{equation}
where $\mathrm{d}$ is the exterior derivative and $|_{\partial M}$ stands for a trace on the boundary $\partial M$. 
Then $\mathcal{D}=\mathcal{D}^\perp$, that is, $\mathcal{D}$ is a Dirac structure.
\end{thm}

\begin{rmk}
Although the differential operator in (\ref{eq-7Cont}), in the presence of nonzero boundary conditions, is not skew-symmetric, it is possible to associate a pseudo-Poisson structure to the Stokes-Dirac structure \cite{vdSM02}. In the absence of algebraic constraints, the Stokes-Dirac structure specializes to a Poisson structure \cite{Dorfman}, and as such it can be derived through symmetry reduction from a canonical Dirac structure on the phase space \cite{JorisCDC}. Whether this reduction can be done for the Stokes-Dirac structure on a manifold with boundary remains an important open problem. 
\end{rmk}

In order to define Hamiltonian dynamics, consider a Hamiltonian density $\mathcal{H}:\Omega^p(M) \times \Omega^q(M)\rightarrow \Omega^n(M)$ resulting with the Hamiltonian $H=\int_M \mathcal{H}\in \mathbb{R}$. Now, consider a time function $t\mapsto (\alpha_p(t),\alpha_q(t))\in \Omega^p(M) \times \Omega^q(M)$, $t\in \mathbb{R}$, and the Hamiltonian $t\mapsto H(\alpha_p(t),\alpha_q(t))$ evaluated along this trajectory, then at any $t$
$$\frac{\textmd{d} H}{\textmd{d} t}=\int_M \delta_p H \wedge \frac{\partial \alpha_p}{\partial t}+\delta_q H\wedge \frac{\partial \alpha_q}{\partial t}\,,$$
where $(\delta_pH,\delta_qH)\in \Omega^{n-p}(M)\times \Omega^{n-q}(M)$ are the (partial) variational derivatives of $H$ at $(\alpha_p,\alpha_q)$.

Setting the flows $f_p=-\frac{\partial \alpha_p}{\partial t}$, $f_q=-\frac{\partial \alpha_q}{\partial t}$ and the efforts $e_p=\delta_p H$, $e_q=\delta_q H$, the distributed-parameter port-Hamiltonian system is defined by the relation
$$\left(  -\frac{\partial \alpha_p}{\partial t},  -\frac{\partial \alpha_q}{\partial t},f_b, \delta_p H,\delta_qH, e_b\right)\in \mathcal{D}\,,~~t\in \mathbb{R}\,.$$
For such a system, it straightaway follows that $\frac{\textmd{d} H}{\textmd{d} t}=\int_{\partial M}e_b\wedge f_b$, expressing the fact that the system is lossless. In other words, the increase in the energy of the system is equal to the power supplied to the system through the boundary $\partial M$.

\section{Fundaments of Discrete Exterior Calculus}\label{Sec3} 
The discrete manifolds we employ are oriented manifold-like simplicial complexes and their circumcentric duals. Typically, these manifolds are simplicial approximations of smooth manifolds. Familiar examples are meshes of triangles embedded in $\mathbb{R}^3$ and tetrahedra obtained by tetrahedrization of a $3$-dimensional manifold. There are many ways to obtain such complexes; however, we do not address the issue of discretization and embedding.

As said, we proceed \emph{ab initio} and mostly our treatment is purely formal, that is without proofs that the discrete objects converge to the continuous ones, though we briefly address the issue of convergence in Section~\ref{Sec:convergence}. By construction of discrete exterior calculus, a number of important geometric structures are preserved and propositions like the Stokes theorem are true by definition. The basic building blocks of discrete exterior geometry are discrete chains and cochains, and their geometric duals. The former are simplices and the latter are discrete differential forms related one to another by bilinear pairing that can be understood as the evaluation of a cochain on an appropriate simplex, and as such parallels the integration in the continuous setting.

In discrete exterior calculus, a dual mesh is instrumental for defining the diagonal Hodge star. In the paper at hand, the geometric duality is a crucial ingredient in establishing a bijective relationship between the flow and effort spaces, as well as for construction of nondegenerate discrete analogues of the bilinear form (\ref{eq-aj9}).

This section, with some modification concerning the treatment of the boundary of the dual cell complex, is a brief summary of the essential definitions and results in discrete exterior calculus as developed in \cite{Desbrun1,Desbrun, Hirani}. As therein, first we define discrete differential forms, the discrete exterior derivative, the codifferential operator, the Hodge star, and the discrete wedge product. For more information on the construction of the other discrete objects such as vector fields, a discrete Lie derivative, and discrete musical operators, we refer the reader to \cite{Hirani}. Construction of all these discrete objects is in a way simpler than their continuous counterparts since we require only a local metric, ergo the machinery of the Riemannian geometry is not demanded. With the exception of the treatment of the notions related the boundary of the dual cell complex, a good part of this section is a recollection of the basic concepts and results of algebraic topology \cite{Hatcher, Munkres}.

\subsection{Simplicial complexes and their circumcentric duals}     

\begin{dfn}
A $k$-simplex is the convex span of $k+1$ geometrically independent points,
\begin{equation*}
\sigma^k =\left[v_0,v_1,\hdots, v_k \right]=\left\{ \sum_{i=0}^k \alpha^i v_i  \,\big |\, \alpha^i \geq  0, \sum_{i=0}^n \alpha^i =1 \right\}.
\end{equation*}
The points $v_0,\ldots,v_k$ are called the vertices of the simplex, and the number $k$ is called the dimension of the simplex. Any simplex spanned by a (proper) subset of $\{v_0,\ldots,v_k\}$ is called a (proper) face of $\sigma^k$. If $\sigma^l $ is a proper face of $\sigma^k$, we denote this by $\sigma^l \prec \sigma^k$.

\end{dfn}

As an illustration, consider four non-collinear points $v_0$, $v_1$, $v_2$, and $v_3$ in $\mathbb{R}^3$. Each of these points individually is $0$-simplex with an orientation dictated by the choice of a sign. An example of a $1$-simplex is a line segment $[v_0,v_1]$ oriented from $v_0$ to $v_1$. The triangle $[v_0,v_1,v_2]$ is an example of $2$-simplex oriented in counterclockwise direction. Similarly, the tetrahedron $[v_0,v_1,v_2,v_3]$ is a $3$-simplex.

\begin{dfn}
A simplicial complex $K$ in $\mathbb{R}^N$ is a collection of simplices in $\mathbb{R}^N$, such that:
\begin{itemize}
  \item[$(1)$] Every face of a simplex of $K$ is in $K$. 
  \item[$(2)$] The intersection of any two simplices of $K$ is a face of each of them.
\end{itemize}
The dimension $n$ of the highest dimension simplex in $K$ is the dimension of $K$. 
\end{dfn}

The above given definition of a simplicial complex is more general than needed for the purposes of exterior calculus. Since the discrete theory employed in this paper mirrors the continuous framework, we restrict our considerations to \emph{manifold-like simplicial complexes} \cite{Hirani}.

\begin{dfn}
A simplicial complex $K$ of dimension $n$ is a manifold-like simplicial complex if the underlying space is a polytope $|K|$. In such a complex all simplices of dimension $k=0,\ldots,n-1$ must be a face of some simplex of dimension $n$ in the complex.
\end{dfn}

Introducing these simplicial meshes has an added advantage of allowing a simple and intuitive definition of orientability of simplicial complexes \cite{Hirani}.

\begin{dfn}
An $n$-dimensional simplicial complex $K$ is an oriented manifold-like simplicial complex if the $n$-simplices that share a common $(n-1)$-face have the same orientation and all the simplices of lower dimensions are individually oriented.
\end{dfn}

Henceforth in this paper, we shall work with manifold-like simplicial complexes. When no confusion can arise, we address these objects simply as simplicial complexes.

An essential constituent of discrete exterior calculus is the dual complex of a manifold-like simplicial complex. The most popular notions of duality are barycentric and circumcentric, also known as Voronoi, duality. Following the standard approach of discrete exterior calculus, in this paper we employ the latter.

The circumcenter of a $k$-simplex $\sigma^k$ is given by the center of the $k$-circumsphere, where the $k$-circumsphere is the unique $k$-sphere that has all $k+1$ vertices of $\sigma^k$ on its surface. That is, the circumcenter is the unique point in the $k$-dimensional affine space that contains the $k$-simplex that is equidistant from all the $k+1$ nodes of the simplex. We denote the circumcenter of a simplex $\sigma^k$ by $c(\sigma^k)$.

If the circumcenter of a simplex lies in its interior we call it a well-centered simplex. For instance, a triangle with all acute angles is a well-centered $2$-simplex. A simplicial complex $K$ whose all simplices of all dimensions are well-centered is called a well-centered simplicial complex and its dual obtained by circumcentric subdivision is also a simplicial complex denoted by $\textmd{csd}\,K$ and its elements by $\hat{\sigma}^0,\ldots,\hat{\sigma}^n$. Throughout this paper, we adopt the convention that all symbols related to the dual (simplicial and cell) complex are labeled by a caret. The underlying spaces $|K|$ and $|\textmd{csd}\,K|$ are the same. The simplicial complex $\textmd{csd} \,K$ consists of all simplices of the form $[c(\sigma_1),\ldots,c(\sigma_k)]$ for $k=1,\ldots,n$, where $\sigma_1\prec\sigma_2\prec\ldots\prec \sigma_k$, meaning $\sigma_i$ is a proper face of $\sigma_j$ for all $i<j$.

A circumcentric dual cell complex (block complex in terminology of \cite{Munkres}) is obtained by aggregation of certain simplices of $\textmd{csd} \,K$. Let $K$ be a well-centered simplicial complex of dimension $n$ and let $\sigma^k$ be one of its simplices. By $D(\sigma^p)$ we denote the union of all open simplices of $\textmd{csd} \,K$ of  which $c(\sigma^k)$ is the final vertex; this cell is the dual cell to $\sigma^k$. The closure of the dual cell of $\sigma^k$ is $\bar{D}(\sigma^k)$. The collection of all dual cells is a cell complex denoted by $D(K)$ with closure $\bar{D}(K)$.

\begin{figure}
\centering
    \includegraphics[width=13cm]{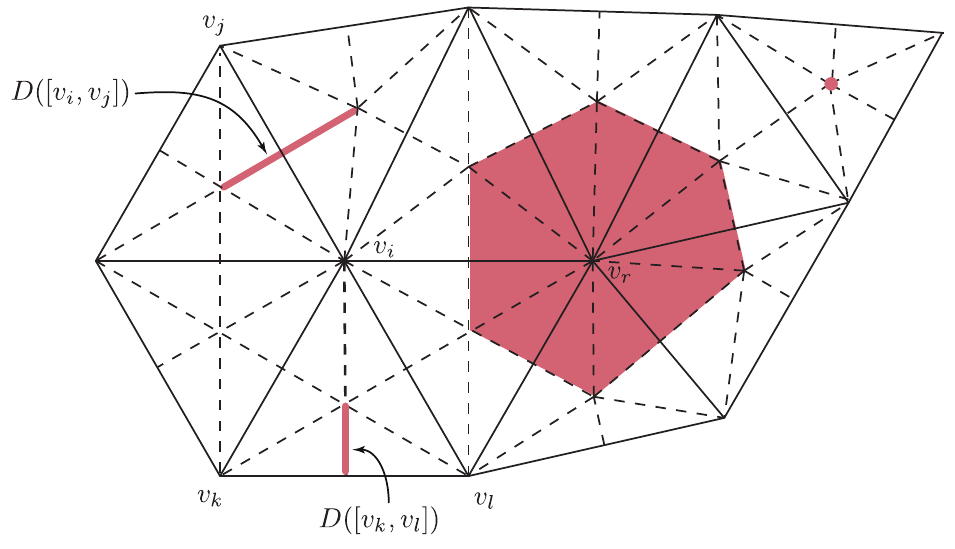}\\
  \caption{A $2$-dimensional simplicial complex $K$ subdivided into the circumcentric simplicial complex $\mathrm{csd}\,K$ indicated by dotted lines. The dual cells displayed are shaded.}\label{fig1}
\end{figure}

To illustrate the duality, consider the $2$-dimensional simplicial complex pictured in Figure~\ref{fig1}. The dual cell of the vertex $v_r$ is the topological interior of the Voronoi region around it as shown shaded in the figure. This dual cell is comprised of the vertex $v_r$, the interior of the open edges emanating from $v_r$, and interiors of the all dual simplices containing $v_r$. The dual cell of any $2$-simplex consists of its circumcenter alone. The dual cell of an edge consists of its circumcenter and two open edges joining this circumcenter to the circumcenters of two triangles having the primal edge as a face. The dual of a boundary edge has only one half-edge since there is only one triangle adjacent to that boundary edge. Note that if the complex is not flat, then the dual edge will not be a straight line.

\begin{rmk}
A triangulation of a compact $n$-dimensional Riemannian manifold $M$ results in an $n$-dimensional simplicial complex $K$. Intuitively, the simplices are glued to the manifold $M$ in such a way that they form a `curved' manifold-like simplicial complex. It is worth noticing that in practical applications, the smooth manifold sometimes is unknown and can only be sampled by physical measurements. In such situations, it makes sense to model the spatial domain as inherently discrete. This is where discrete port-Hamiltonian theory in the framework of discrete exterior calculus stands in its own right. 
\end{rmk}

\subsection{Chains and cochains}
The discrete analogue of a smooth $k$-form is a $k$-cochain, a certain type of a function, on a $k$-chain representing a formal sum of simplices. The role of integration in the discrete theory is replaced by (simple) evaluation of a discrete form on a chain. The discrete exterior derivative is defined by duality to the boundary operator, rendering the Stokes theorem true by definition. Parallel to the smooth case, the discrete exterior wedge product pairs lower degree forms into a higher degree one.

\begin{dfn}
Let $K$ be a simplicial complex. We denote the free Abelian group generated by a basis consisting of oriented $k$-simplices by $C_k(K;\mathbb{Z})$. This is the space of finite formal sums of the $k$-simplices with coefficients in $\mathbb{Z}$. Elements of $C_k(K;\mathbb{Z})$ are called $k$-chains.
\end{dfn}

\begin{dfn}
A primal discrete $k$-form $\alpha$ is a homomorphism from the chain group $C_k(K;\mathbb{Z})$ to the additive group $\mathbb{R}$. Thus, a discrete $k$-form is an element of $\mathrm{Hom}(C_k(K),\mathbb{R})$, the space of cochains. This space becomes an Abelian group if we add two homomorphisms by adding their values in $\mathbb{R}$. The standard notation for $\mathrm{Hom}(C_k(K),\mathbb{R})$ in algebraic topology is $C^k(K;\mathbb{R})$; however, like in \cite{Desbrun1,Desbrun, Hirani} we shall also employ the notation $\Omega_d^k(K)$ for this space as a reminder that this is the space of discrete $k$-forms on the simplicial complex $K$. Thus,
\begin{equation*}
\Omega_d^k(K):=C^k(K;\mathbb{R})=\mathrm{Hom}(C_k(K),\mathbb{R})\,.
\end{equation*}

\end{dfn}

Given a $k$-chain $\sum_i a_i c_i^k$, $a_i \in \mathbb{Z}$, and a discrete $k$-form $\alpha$, we have
\begin{equation*}
\alpha\left( \sum_i a_i c_i^k \right)=\sum_i a_i \alpha(c_i^k)\,,
\end{equation*}
and for two discrete $k$-forms $\alpha,\beta\in \Omega_d^k(K)$ and a $k$-chain $c\in C_k(K;\mathbb{Z})$,
\begin{equation*}
(\alpha+\beta)(c)=\alpha(c)+\beta(c)\,.
\end{equation*}

The natural pairing of a $k$-form $\alpha$ and a $k$-chain $c$ is defined as the bilinear pairing
$\langle \alpha,c \rangle =\alpha(c)$.

As previously pointed out, a differential $k$-form $\alpha^k$ can be thought of as a linear functional that assigns a real number to each oriented cell $\sigma^k\in K$. In order to understand the process of discretization of the continuous problem consider a smooth $k$-form $f\in \Omega^k(|K|)$. The discrete counterpart of $f$ on a $k$-simplex $\sigma^k\in K$ is a discrete form $\alpha^k$ defined as $\alpha(\sigma^k):=\int_{\sigma^k}f$.

\begin{dfn}
The boundary operator $\partial_k:C_k(K;\mathbb{Z})\rightarrow C_{k-1}(K;\mathbb{Z})$ is a homomorphism defined by its action on a simplex $\sigma^k=[v_0,\ldots,v_k]$,
\begin{equation*}
\partial_k \sigma^k=\partial_k ([v_0,\ldots,v_k])=\sum_{i=0}^k(-1)^i [v_0,\ldots,\hat{v}_i,\ldots,v_k]\,,
\end{equation*}
where $[v_0,\ldots,\hat{v}_i,\ldots,v_k]$ is the $(k-1)$-simplex obtained by omitting the vertex $v_i$. Note that $\partial_k\circ \partial_{k+1}=0$.
\end{dfn}

\begin{dfn}
On a simplicial complex of dimension $n$, a chain complex is a collection of chain groups and homomorphisms $\partial_k$, such that
\begin{equation*}
\begin{split}
0{\longrightarrow}C_n(K) \xrightarrow{\partial_{n}}
C_{n-1}\xrightarrow{\partial_{n-1}}\cdots
\xrightarrow{\partial_{k+1}} C_k(K)\xrightarrow{\partial_{k}}\cdots
\xrightarrow{\partial_{1}}C_0(K)\xrightarrow{\partial_{0}}0\,,
\end{split}
\end{equation*}
and $\partial_k\circ \partial_{k+1}=0$.
\end{dfn}

\begin{dfn}
The discrete exterior derivative $\mathbf{d}:\Omega_d^k(K)\rightarrow \Omega_d^{k+1}(K)$ is defined by duality to the boundary operator $\partial_{k+1}:C_{k+1}(K;\mathbb{Z})\rightarrow C_{k}(K;\mathbb{Z})$, with respect to the natural pairing between discrete forms and chains. For a discrete form $\alpha^k\in \Omega_d^k(K) $ and a chain $c_{k+1}\in C_{k+1}(K;\mathbb{Z})$ we define $\mathbf{d}$ by
\begin{equation*}
\langle \mathbf{d} \alpha^k,c_{k+1}\rangle=\langle \alpha^k,\partial_{k+1}c_{k+1}\rangle\,.
\end{equation*}
The discrete exterior derivative is the coboundary operator from algebraic topology \cite{Munkres} and as such it induces the cochain complex
\begin{equation*}
\begin{split}
0\,{\longleftarrow}\Omega_d^n(K) \xleftarrow{~\mathbf{ d}~}
\Omega_d^{n-1}\xleftarrow{~\mathbf{ d}~}\cdots
\xleftarrow{~\mathbf{ d}~}\Omega_d^0(K)\stackrel{\,}{\longleftarrow}0\,,
\end{split}
\end{equation*}
where $\mathbf{ d}\circ \mathbf{ d}=0$.

\end{dfn}

Such as in the continuous theory, we drop the index of the boundary operator when its dimension is clear from the context. The discrete exterior derivative $\mathbf{d}$ is constructed in such a manner that the Stokes theorem is satisfied by definition. This means, given a $(k+1)$-chain $c$ and a discrete $k$-form $\alpha$, the {\emph{discrete Stokes theorem}} states that
\begin{equation*}
\langle \mathbf{d} \alpha,c\rangle=\langle \alpha,\partial c\rangle\,.
\end{equation*}

Consider a $k$-chain $\sum_i {a_i}c_i$, $a_i\in \mathbb{Z}$, $c_i\in C_k(K;\mathbb{Z})$, and $(k-1)$-form $\alpha\in \Omega_d^{k-1}(K;\mathbb{Z})$. By linearity of the chain-cochain pairing, the discrete Stokes theorem can be stated as
\begin{equation*}
\left \langle \mathbf{d} \alpha,\sum_i a_i c_i\right \rangle=\left \langle \alpha,\partial\left(\sum_i a_i c_i\right)\right \rangle=\left \langle \alpha,\sum_i a_i \partial c_i\right \rangle=\sum_i a_i  \left \langle \alpha, \partial c_i\right \rangle\,.
\end{equation*}

As in the continuous setting, the discrete wedge product pairs two discrete differential forms by building a higher degree form. The primal-primal wedge product inherits some important properties of the cup product such as the bilinearity, anticommutativity and naturality under pull-back \cite{Desbrun1,Hirani}; however, it is in general non-associative and degenerate, and thus unsuitable for construction of canonical pairing between the flow and effort space. For a definition of nondegenerate pairing between the flow and effort discrete forms we shall use a primal-dual wedge product as will be defined in the subsequent section.

\subsection{Metric-dependent part of discrete exterior calculus}
A cellular chain group associated with the dual cell complex $D(K)$, in \cite{Munkres} denoted by $D_p(K)$, is the group of formal sums of cells with integer coefficients. Since in $D(K)$ the information of dual simplices is lost, to retain the bookkeeping information Hirani in \cite{Hirani} introduces a duality operator which takes values in the domain group $C_p(\textmd{csd} \,K; \mathbb{Z})$. As will be clear from the subsequent section, this bookkeeping is not indispensable for the formulation of the Dirac structure on a simplicial complex; nevertheless, since the information of dual simplices might be needed in defining dynamics, we also employ this construction.

In order to explicitly construct the duality on the boundary, in the next definition we introduce the boundary star operator. Shortly afterward we shall explain the rational behind this construction. 
\begin{dfn}
Let $K$ be a well-centered simplicial complex of dimension $n$. The interior circumcentric duality operator $\star_\mathrm{i}  :C_k(K;\mathbb{Z})\rightarrow C_{n-k}({\mathrm{csd}} \,K;\mathbb{Z})$
\begin{equation*}
\star_\mathrm{i} (\sigma^k)=\sum_{\sigma^k\prec\sigma^{k+1}\prec \cdots\prec \sigma^n} s_{\sigma^k,\ldots,\sigma^n}\left[ c(\sigma^k), c(\sigma^{k+1}),\ldots, c(\sigma^n) \right]\,,
\end{equation*}
and the boundary star operator $\star_\mathrm{b} :C_k(\partial K;\mathbb{Z})\rightarrow C_{n-1-k}(\partial (\mathrm{csd} \,K);\mathbb{Z})$
\begin{equation*}
\star_\mathrm{b} (\sigma^k)=\sum_{\sigma^k\prec\sigma^{k+1}\prec \cdots\prec \sigma^{n-1}} s_{\sigma^k,\ldots,\sigma^{n-1}}\left[ c(\sigma^k), c(\sigma^{k+1}),\ldots, c(\sigma^{n-1}) \right]\,,
\end{equation*}
where the $s_{\sigma^k,\ldots,\sigma^n}$ and $s_{\sigma^k,\ldots,\sigma^{n-1}}$ coefficients ensure that the orientation of the cell $[ c(\sigma^k),\\c(\sigma^{k+1}),\ldots, c(\sigma^n) ]$ and $[ c(\sigma^k), c(\sigma^{k+1}),\ldots, c(\sigma^{{n-1}}) ]$  is consistent with the orientation of the primal simplex, and the ambient volume forms on $K$ and $\partial K$, respectively.
\end{dfn}

The subset of chains $C_p(\mathrm{csd}\,K;\mathbb{Z})$ that are equal to the cells of $D(K)\times D(\partial K)$ forms a subgroup of $C_p(\mathrm{csd}\,K;\mathbb{Z})$. We denote this subgroup of $C_p(\mathrm{csd}\,K;\mathbb{Z})$ by $C_p(\star K;\mathbb{Z})$, where $\star K$ is its basis set. A cell complex $\star K$ in $\mathbb{R}^N$ is a collection of cells in $\mathbb{R}^N$ such that: $(1)$ there is a partial ordering of cells in $\star K$, $\hat{\sigma}^{k}\prec \hat{\sigma}^l$, which is read as $\hat{\sigma}^k$ is a face of $\hat{\sigma}^l$; $(2)$ the intersection of any two cells in $\star K$, is either a face of each of them, or it is empty; $(3)$ the boundary of a cell is expressed as a sum of its proper faces. 

Given a simplicial well-centered complex $K$, we define its interior dual cell complex $\star_\mathrm{i} K$ (block complex in terminology of algebraic topology \cite{Munkres}) as a circumcentric dual restricted to $|K|$. An important property of the the Voronoi duality is that primal and dual cells are orthogonal to each other. The boundary dual cell complex $\star_\mathrm{b} K$ is a dual to $\partial K$. The dual cell complex $\star K$ is defined as $\star K=\star_\mathrm{i} K \times \star_\mathrm{b} K$. A dual mesh $\star_\mathrm{i} K$ is a dual to $K$ in sense of a graph dual, and the dual of the boundary is equal to the boundary of the dual, that is $\partial (\star K)=\star (\partial K)=\star_\mathrm{b} K$. This construction of the dual is compatible with \cite{Stern,Hiptmair1} and as such is very similar to the use of the ghost cells in finite volume methods in order to account for the duality relation between the Dirichlet and the Neumann boundary conditions. Because of duality, there is a one-to-one correspondence between $k$-simplices of $K$ and interior $(n-k)$-cells of $\star K$. Likewise, to every $k$-simplex on $\partial K$ there is a uniquely associated $(n-1-k)$-cell on $\partial(\star K)$.

\begin{figure}
\centering
    \includegraphics[width=13cm]{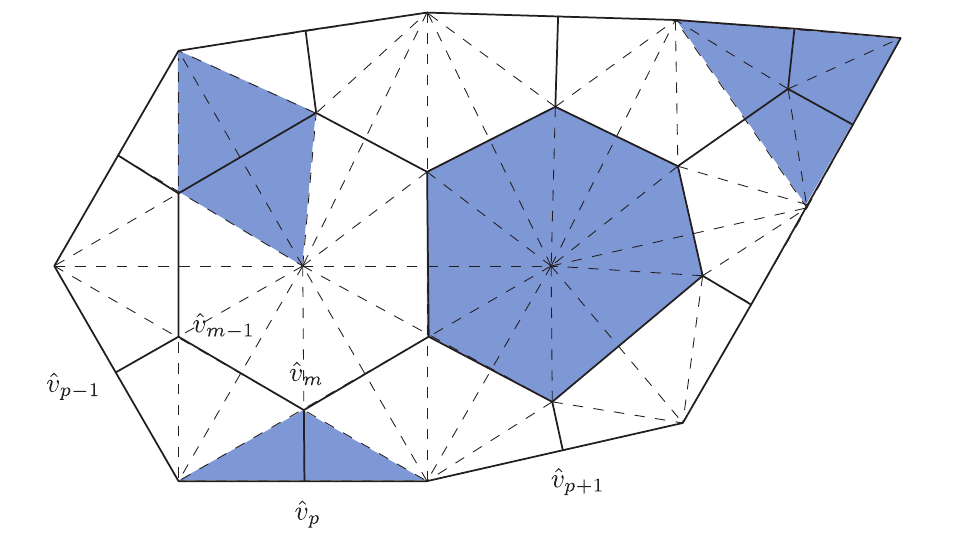}\\
  \caption{The circumcentric dual cell complex $\star K$ of the simplicial complex $K$ given in Figure~\ref{fig1}. The boundary of $\star K$ is the dual of the boundary of $K$. Some support volumes are shaded. For instance, the support volume of the primal vertex $v_r$ is the area of its Voronoi region; also $V_{[v_i,v_j]}=V_{\bar{D}([v_i,v_j])}$.}\label{fig2}
\end{figure}

In what follows, we shall abuse notation and use the same symbol $\star$ for both the interior circumcentric and the boundary star operator. The difference, if not clear from the exposition, will be delineated by indicating that $\star \sigma^k\in \partial (\star K)$ when $\star \sigma^k$ is a dual cell on the boundary of the dual cell complex $\star K$. As sets, the set of $\bar{D}(\sigma^p)$ and $\star \sigma^p$ are equal, the only difference being in the  bookkeeping, since in $\star \sigma^p$ one retains the information about the simplices it is built of. Here we do not address the problem of the orientation of dual $\star K$, for which we direct the reader to \cite{Hirani}. The circumcentric dual cell complex of the $2$-dimensional simplicial complex from Figure~\ref{fig1} is pictured in Figure~\ref{fig2}.

An important concept in defining a wedge product between primal and dual cochains is the notion of a support volume associated with a given simplex or cell. 
\begin{dfn}
The support volume of a simplex $\sigma^k$ is an $n$-volume given by the convex hull of the geometric union of the simplex and its circumcentric dual. This is given by
\begin{equation*}
V_{\sigma^k}=\mathrm{CH}(\sigma^k,\star_\mathrm{i} \sigma^k)\cap |K|\,,
\end{equation*}
where $\mathrm{CH}(\sigma^k,\star_\mathrm{i} \sigma^k)$ is the $n$-dimensional convex hull generated by $\sigma^k\cup \star_\mathrm{i} \sigma^k$. The intersection with $|K|$ is necessary to ensure that the support volume does not extend beyond the polytope $|K|$ which would otherwise occur if $K$ is nonconvex.
The support volume of a dual cell $\star_\mathrm{i}\sigma^k$ is 
\begin{equation*}
V_{\star_\mathrm{i} \sigma^k}=\mathrm{CH}(\star_\mathrm{i}\sigma^k,\star_\mathrm{i}\star_\mathrm{i} \sigma^k)\cap |K|=V_{\sigma^k}\,.
\end{equation*}
\end{dfn}

Everything that has been said about the primal chains and cochains can be extended to dual cells and dual cochains. We do not elaborate on this since it can be found in the literature \cite{Hirani, Desbrun}, however, in order to properly account for the behaviours on the boundary, we need to adapt the definition of the boundary dual operator as presented in \cite{Hirani, Desbrun}. We propose the following definition.

\begin{dfn}
The dual boundary operator $\partial_k: C_k(\star_\mathrm{i} K;\mathbb{Z})\rightarrow C_{k-1}(\star K;\mathbb{Z})$ is a homomorphism defined by its action on a dual cell $\hat{\sigma}^k=\star_\mathrm{i} \sigma^{n-k}=\star_\mathrm{i}[v_0,\ldots,v_{n-k}]$,
\begin{equation*}
\begin{split}
\partial \hat{\sigma}^k=\partial \star_\mathrm{i} [v_0,\ldots,v_{n-k}]&=\partial_\mathrm{i} \star_\mathrm{i} [v_0,\ldots,v_{n-k}]+\partial_\mathrm{b} \star_\mathrm{i} [v_0,\ldots,v_{n-k}]\,,
\end{split}
\end{equation*}
where
\begin{equation*}
\begin{split}
\partial_\mathrm{i} \star_\mathrm{i} [v_0,\ldots,v_{n-k}]&=\sum_{\sigma^{n-k+1}\succ \sigma^{n-k}}\star_\mathrm{i} (s_{\sigma^{n-k+1}}\sigma^{n-k+1})\\
\partial_\mathrm{b} \star_\mathrm{i} [v_0,\ldots,v_{n-k}]&=\star_\mathrm{b} \left(s_{\sigma^{n-k}}\sigma^{n-k} \right)\,.
\end{split}
\end{equation*}

\end{dfn}
Note that the dual boundary operator as defined in \cite{Hirani} is equal to $\partial_\mathrm{i}$. Hence, the dual boundary is not the geometric boundary of a cell, because near the boundary of a manifold that would be wrong. As an example consider the complex in Figure~\ref{fig3}. The dual of the vertex $v_1$ is the Voronoi region shown shaded. Its geometric boundary has five sides (two half primal edges and three dual edges), whereas the dual boundary according to the definition given in \cite{Hirani} consists of just dual edges, i.e. $[c([v_0,v_1]),c([v_0,v_1,v_2])]$, $[c([v_0,v_1,v_2]),c([v_1,v_3,v_2])]$ and $[c([v_1,v_3,v_2]),c([v_1,v_3])]$, all up to a sign depending on the chosen orientation. However, according to the above given definition, the boundary is comprised of four edges, three already given plus the boundary edge $[c([v_0,v_1]),c([v_1,v_3])]$ obtained by aggregation of the two dual simplices $[c([v_0,v_1]),v_1]$ and $[v_1,c([v_1,v_3])]$ . This construction of the dual boundary ensures a natural pairing between a primal $0$-form defined on $v_1$ and a dual $1$-form on $[c([v_0,v_1]),c([v_1,v_3])]$. The offered definition of the dual boundary operator, as will be demonstrated later, is crucial for the inclusion of the boundary variables in the discrete setting.

\begin{dfn}
The dual discrete exterior derivative $\mathbf{d}:\Omega_d^k(\star K)\rightarrow \Omega_d^{k+1}(\star_\mathrm{i} K)$ is defined by duality to the boundary operator $\partial :C_{k+1}(\star_\mathrm{i} K;\mathbb{Z})\rightarrow C_{k}(K;\mathbb{Z})$. For a dual discrete form $\hat{\alpha}^k\in \Omega_d^k(\star K) $ and a chain $\hat{c}_{k+1}\in C_{k+1}(\star_\mathrm{i} K;\mathbb{Z})$ we define $\mathbf{d}$ by
\begin{equation*}
\langle \mathbf{d} \hat{\alpha}^k,\hat{c}_{k+1}\rangle=\langle \hat{\alpha}^k,\partial \hat{c}_{k+1}\rangle\,.
\end{equation*}
The dual discrete exterior derivative $\mathbf{d}$ can be decomposed into the two operators $\mathbf{d}_\mathrm{i}$ and $\mathbf{d}_\mathrm{b}$, which are respectively dual to $\partial_\mathrm{i}$ and $\partial_\mathrm{i}$, that is
\begin{equation*}
\langle \mathbf{d} \hat{\alpha}^k,\hat{c}_{k+1}\rangle=\langle \mathbf{d}_\mathrm{i }\hat{\alpha}^k,\hat{c}_{k+1}\rangle+\langle \mathbf{d}_\mathrm{b} \hat{\alpha}^k,\hat{c}_{k+1}\rangle=\langle \hat{\alpha}^k,\partial_\mathrm{i }\hat{c}_{k+1}\rangle+\langle \hat{\alpha}^k,\partial_\mathrm{b} \hat{c}_{k+1}\rangle\,.
\end{equation*}
\end{dfn}

\begin{figure}
\centering
    \includegraphics[width=4.5cm]{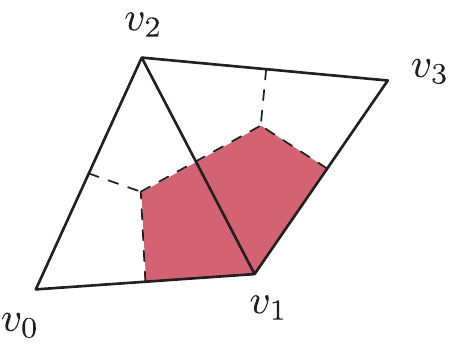}\\
  \caption{A $2$-dimensional simplicial complex taken from Figure~3.3, Section 3.6, \cite{Hirani}. The shaded region is a Voronoi dual of the primal vertex $v_1$. The dual boundary, according to \cite{Hirani}, is not the geometric boundary near the boundary of the manifold. However, in line with our construction, the boundary of the dual is the dual of the boundary.}\label{fig3}
\end{figure}

The support volumes of a simplex and its dual cell are the same, which suggests that there is a natural identification between primal $k$-cochains and dual $(n-k)$-cochains. 

In the exterior calculus for smooth manifolds, the Hodge star, denoted $*$, is an isomorphism between the space of $k$-forms and $(n-k)$-forms. Since the Hodge star operator is metric-dependent, in the discrete theory, it is defined as an equality of averages between primal and their dual forms \cite{Hirani,Hiptmair}.

\begin{dfn}\label{df:HodgeS}
The discrete Hodge star is a map $* : \Omega_d^k(K)\rightarrow \Omega_d^{n-k}(\star_\mathrm{i} K)$ defined by its value over simplices and their duals. For a $k$-simplex $\sigma^k$, and a discrete $k$-form $\alpha^k$,
\begin{equation*}
\frac{1}{|\star_\mathrm{i} \sigma^k|}\langle * \alpha^k,\star_\mathrm{i}\sigma^k \rangle=s \frac{1}{|\sigma^k|}\langle \alpha^k,\sigma^k \rangle\,,
\end{equation*}
where $s$ is $\pm 1$ (see \cite{Hirani}).
\end{dfn}

Similarly we can define the discrete Hodge operator on the boundary, that is on an $(n-1)$-dimensional simplicial complex and its dual. It is trivial to show that for a $k$-form $\alpha^k$ the following holds: $**\alpha^k=(-1)^{k(n-k)}$.

\begin{rmk}
The discrete Hodge star can be represented by a matrix (see Section~\ref{Sec:matrixrep}). According to Definition~\ref{df:HodgeS}, this matrix is diagonal. In the case Whitney forms are used, the discrete Hodge operator is sparse but not diagonal in general \cite{Hiptmair}.
\end{rmk}

Next, we define a natural pairing, via so-called primal-dual wedge product, between a primal $k$-cochain and a dual $(n-k)$-cochain. The resulting discrete form is the volume form. In order to insure anticommutativity of the primal-dual wedge product, we take the following definition.
\begin{dfn}
Let $\alpha^k\in \Omega_d^k(K)$ be a primal $k$-form and $\hat{\beta}^{n-k}\in \Omega_d^{n-k}(\star_\mathrm{i} K)$. We define the discrete primal-dual wedge product $\wedge: \Omega_d^k(K)\times \Omega_d^{n-k}(\star_\mathrm{i} K) \rightarrow \Omega_d^{n}(V_k(K))$ by
\begin{equation*}
\begin{split}
\langle \alpha^k\wedge \hat{\beta}^{n-k},V_{\sigma^{k}}\rangle&=\binom{n}{k} \frac{|V_{\sigma^k}|}{|\sigma^k| |\star_\mathrm{i} \sigma^k|}\langle \alpha^k,\sigma^k \rangle \langle \hat{\beta}^{n-k}, \star_\mathrm{i} \sigma^k\rangle\\
&=\langle \alpha^k,\sigma^k \rangle \langle \hat{\beta}^{n-k}, \star_\mathrm{i}\sigma^k\rangle\\
&=(-1)^{k(n-k)}\langle \hat{\beta}^{n-k}\wedge \alpha^k,V_{\sigma^{k}}\rangle\,,
\end{split}
\end{equation*}
where $V_{\sigma^k}$ is the $n$-dimensional support volume obtained by taking the convex hull of the simplex $\sigma^k$ and its dual $\star_\mathrm{i} \sigma^k$.
\end{dfn}

As an illustration, consider the two-dimensional simplicial complex $K$ depicted in Figure~\ref{fig1} and its dual $\star K$ in Figure~\ref{fig2}. Let $\alpha^1\in \Omega_d^1(K)$ and $\hat{\beta}^1\in \Omega_d^1(\star K)$. The dual cell of the primal edge $[v_k,v_i]$ is $[\hat{v}_m,\hat{v}_{m-1}]$, up to a sign depending on the chosen orientation. The primal-dual wedge product of $\alpha^1$ and $\hat{\beta}^1$ on the support volume $V_{[v_k,v_i]}=V_{[\hat{v}_m,\hat{v}_{m-1}]}$, represented by the diamond shaped region generated by $[v_k,v_i]$ and $[\hat{v}_m,\hat{v}_{m-1}]$, is simply a dot product $\alpha([v_k,v_i])\cdot\hat{\beta}^1([\hat{v}_m,\hat{v}_{m-1}])$. Now, in order to look at the primal-dual wedge product on the boundary, let $\gamma^0\in \Omega_d^0(\partial K)$ and $\hat{\eta}^1\in \Omega_d^1(\partial (\star K))$. For instance, $\hat{\eta}^1$ can be a restriction of $\hat{\beta}^1$ on the boundary $\partial (\star K)$. The primal-dual wedge product $\langle \gamma^0\wedge \hat{\eta}^1,V_{v_k}\rangle=\gamma^0(v_k)\cdot \hat{\eta}^1([\hat{v}_p,\hat{v}_{p-1}])$. The volume for $V_{v_k}=V_{[\hat{v}_p,\hat{v}_{p-1}]}$ is simply the measure of the cell $[\hat{v}_p,\hat{v}_{p-1}]$.

Here we note the advantage of employing circumcentric with respect to barycentric dual since one needs to store only volume information about primal and dual cells, and not about the primal-dual convex hulls.

\begin{dfn}
Given two primal discrete $k$-forms, $\alpha^k,\beta^k\in \Omega_d^k(K)$, their discrete $L^2$ inner product, $\langle \alpha^k,\beta^k \rangle_d$ is given by
\begin{equation*}
\begin{split}
\langle \alpha^k, \beta^{k}\rangle_d&=\binom{n}{k} \frac{|V_{\sigma^k}|}{|\sigma^k| |\star_\mathrm{i} \sigma^k|}\langle \alpha^k,\sigma^k \rangle \langle *{\beta}^{k}, \star_\mathrm{i} \sigma^k\rangle\\
&=\langle \alpha^k,\sigma^k \rangle \langle *{\beta}^{k}, \star_\mathrm{i}\sigma^k\rangle\,.
\end{split}
\end{equation*}
\end{dfn}

The proposed definition of the dual boundary operator assures the validity of the summation by parts relation that parallels the integration by parts formula for smooth differential forms. 

\begin{prop}\label{prop:prim-dual}
Let $K$ be an oriented well-centered simplicial complex. Given a primal $(k-1)$-form $\alpha^{k-1}$ and a dual $(n-k)$-discrete form $\hat{\beta}^{n-k}$, then
\begin{equation*}
\begin{split}
\langle \mathbf{d} \alpha^{k-1}\wedge \hat{\beta}^{n-k},K\rangle+(-1)^{k-1}\langle \alpha^{k-1}\wedge  \mathbf{d} \hat{\beta}^{n-k},K\rangle =\langle  \alpha^{k-1}\wedge \hat{\beta}^{n-k},\partial K\rangle\,,
\end{split}
\end{equation*}
where in the boundary pairing $\alpha^{k-1}$ is a primal $(k-1)$-form on $\partial K$, while $ \hat{\beta}^{n-k}$ is a dual $(n-k)$-cochain taken on the boundary dual $\star (\partial K)$.
\end{prop}

\begin{proof}
We have
\begin{equation*}
\begin{split}
\langle \mathbf{d}\alpha^{k-1}\wedge\hat{\beta}^{n-k},K\rangle&=\sum_{\sigma^{k-1}\in K}\langle \mathbf{d} \alpha^{k-1},\sigma^{k}\rangle \langle \hat{\beta}^{n-k},\star_\mathrm{i} \sigma^{n-k} \rangle\\
&=\sum_{\sigma^{k-1}\in K}\langle  \alpha^{k-1},\partial \sigma^{k}\rangle \langle \hat{\beta}^{n-k},\star_\mathrm{i} \sigma^{n-k} \rangle\\
&=\sum_{\sigma^{k-1}\in K}\sum_{\sigma^{k-1}\prec\sigma^{k}}\langle  \alpha^{k-1}, \sigma^{k}\rangle \langle \hat{\beta}^{n-k},\star_\mathrm{i} \sigma^{n-k} \rangle\,,
\end{split}
\end{equation*}
and 
\begin{equation*}
\begin{split}
\langle \alpha^{k-1}&\wedge  \mathbf{d} \hat{\beta}^{n-k},K\rangle=\sum_{\sigma^{k-1}}\langle \alpha^{k-1},\sigma^{k-1}\rangle \langle \mathbf{d} \hat{\beta}^{n-k},\star_\mathrm{i} \sigma^{k-1} \rangle=\sum_{\sigma^{k-1}}\langle \alpha^{k-1},\sigma^{k-1}\rangle \langle \hat{\beta}^{n-k}, \partial(\star_\mathrm{i}\sigma^{k-1}) \rangle\\
&=\sum_{\sigma^{k-1}}\langle \alpha^{k-1},\sigma^{k-1}\rangle \left(\sum_{\sigma^{k-1}\prec\sigma^k}\langle \hat{\beta}^{n-k}, \star_\mathrm{i} (s_{\sigma^{k}}\sigma^{k})\rangle+ \langle \hat{\beta}^{n-k}, \star_\mathrm{b} (s_{\sigma^{k-1}}\sigma^{k-1})\rangle\right)\,.
\end{split}
\end{equation*}
Inducing the orientation of the dual such that $s_{\sigma^{k}}=s_{\sigma^{k-1}}=(-1)^k$ completes the proof.
\end{proof}

\begin{rmk}\label{rmkdecom}Decomposing the dual form $\hat{\beta}^{n-k}$ into the internal and the boundary part as
$\hat{\beta}^{n-k}=\bigg\{ \begin{array}{l} \hat{\beta}_\mathrm{i}\in \Omega_d^{n-k}(\star_\mathrm{i}K) ~\,\mathrm{on}~\star_\mathrm{i}\!K \\ \hat{\beta}_\mathrm{b}\in \Omega_d^{n-k}(\star_\mathrm{b}K) ~{\mathrm{on}}~\partial(\star K)\end{array}$ and decomposing the dual exterior derivative in the same manner, the summation by parts formula can be written as
\begin{equation}\label{eq:sumbypart}
\begin{split}
\langle \mathbf{d} \alpha^{k-1}\wedge \hat{\beta}_\mathrm{i} ,K\rangle&+(-1)^{k-1}\langle \alpha^{k-1}\wedge ( \mathbf{d}_\mathrm{i} \hat{\beta}_\mathrm{i} +\mathbf{d}_\mathrm{b} \hat{\beta}_\mathrm{b} ),K\rangle =\langle  \alpha^{k-1}\wedge \hat{\beta}_\mathrm{b} ,\partial K\rangle\,.
\end{split}
\end{equation}
\end{rmk}

In the standard literature of discrete exterior calculus, the codifferential operator is adjoint to the discrete exterior derivative, with respect to the inner products of discrete forms \cite{Hirani}. According to Proposition~\ref{prop:prim-dual}, this is not the case since on the right a term corresponding to the primal-dual pairing on the boundary appears. As the subsequent section demonstrates, this term is precisely responsible for the inclusion of the boundary variables in the discretized Stokes-Dirac structure.

\section{Dirac Structures on a Simplicial Complex}\label{Sec4}
The Stokes-Dirac structure, which captures a differential symmetry of the Hamiltonian field equations, as presented in \cite{vdSM02}, is metric-independent. The essence of its construction lies in the antisymmetry of the wedge product and the Stokes theorem. In a discrete framework, the primal-primal wedge product \cite{Hirani} inherits a number of important properties of the cup product \cite{Munkres}, such as bilinearity, anti-commutativity and naturality under pullback; however, it is degenerate and thus unsuitable for defining a Dirac structure. This motivates a formulation of a Dirac structure on a simplicial complex and its dual. We introduce Dirac structures with respect to the bilinear pairing between primal and duals forms on the underlying discrete manifold. We call these Dirac structures \emph{simplicial Dirac structures}.

In the discrete setting, the smooth manifold $M$ is replaced by an $n$-dimensional well-centered oriented manifold-like simplicial complex $K$. The flow and the effort spaces will be the spaces of complementary primal and dual forms. The elements of these two spaces are paired via the discrete primal-dual wedge product. Since the Stokes-Dirac structure $\mathcal{D}$ expresses the coupling between $f_p$ and $e_q$, also $f_q$ and $e_p$, via the exterior derivative, whose discrete analogue maps primal into primal and dual into dual cochains, the flow space cannot be entirely built on a primal simplicial complex and the effort space on a dual cell complex, or vice versa. Instead, the flow and the effort spaces will be mixed spaces of the primal and dual cochains. One of the two possible choices is
\begin{equation*}
\mathcal{F}_{p,q}^d=\Omega_d^p(\star_\mathrm{i} K)\times \Omega_d^q( K)\times \Omega_d^{n-p}(\partial (K))\,
\end{equation*}
and 
\begin{equation*}
\mathcal{E}_{p,q}^d=\Omega_d^{n-p}( K)\times \Omega_d^{n-q}(\star_\mathrm{i} K)\times \Omega_d^{n-q}(\partial (\star K))\,.
\end{equation*}

The primal-dual wedge product ensures a bijective relation between the primal and dual forms, between the flows and efforts. A natural discrete mirror of the bilinear form (\ref{eq-aj9}) is a symmetric pairing on the product space $\mathcal{F}_{p,q}^d\times \mathcal{E}_{p,q}^d$ defined by
\begin{equation}\label{eq:bild}
\begin{split}
\langle\!\langle (&\underbrace{\hat{f}_p^1,{f}_q^1,{f}_b^1}_{\in \mathcal{F}_{p,q}^d},\underbrace{{e}_p^1,\hat{e}_q^1,\hat{e}_b^1}_{\in \mathcal{E}_{p,q}^d}), (\hat{f}_p^2,{f}_q^2,{f}_b^2,{e}_p^2,\hat{e}_q^2,\hat{e}_b^2)\rangle\!\rangle_d\\
&= \langle {e}_p^1\wedge \hat{f}_p^2+\hat{e}_q^1\wedge {f}_q^2+ {e}_p^2\wedge \hat{f}_p^1+\hat{e}_q^2\wedge {f}_q^1,K \rangle+\langle \hat{e}_b^1\wedge {f}_b^2+ \hat{e}_b^2\wedge f_b^1,\partial K \rangle
  \,.
  \end{split}
\end{equation}
A discrete analogue of the Stokes-Dirac structure is the finite-dimensional Dirac structure constructed in the following theorem.

\begin{thm}\label{th:pdDirac}
Given linear spaces $\mathcal{F}_{p,q}^d$ and $\mathcal{E}_{p,q}^d$, and the bilinear form $\langle\!\langle, \rangle\!\rangle_d$. The linear subspace $\mathcal{D}_d\subset\mathcal{F}_{p,q}^d\times \mathcal{E}_{p,q}^d$ defined by
\begin{equation}\label{eq:Dir-prim-dual}
\begin{split}
& \mathcal{D}_{d}=\big\{ (\hat{f}_p, {f}_q, {f}_b,{e}_p,\hat{e}_q,\hat{e}_b)\in \mathcal{F}_{p,q}^d\times \mathcal{E}_{p,q}^d\big |\\
& \left(\begin{array}{c}\hat{f}_p \\ {f}_q \end{array} \right) =\left(\begin{array}{cc} 0  & (-1)^{pq+1}  \mathbf{d}_\mathrm{i}\\  \mathbf{d} & 0\end{array} \right) \left(\begin{array}{c} {e}_p \\ \hat{e}_q\end{array} \right)+ (-1)^{pq+1} \left(\begin{array}{c}  \mathbf{d}_\mathrm{b} \\ 0 \end{array}\right) \hat{e}_b\,,\\
& \begin{array}{c}~~~~~f_b \end{array}= ~(-1)^{p}{e}_p|_{\partial K}\,\big \}\,
 \end{split}
\end{equation}
is a Dirac structure with respect to the pairing $\langle\!\langle, \rangle\!\rangle_{d}$ .
\end{thm}

\begin{proof}
In order to show that $\mathcal{D}_d\subset\mathcal{D}_{d} ^\perp$, let $(\hat{f}_p^1,{f}_q^1,{f}_b^1,{e}_p^1,\hat{e}_q^1,\hat{e}_b^1)\in \mathcal{D}_d$, and consider any $(\hat{f}_p^2,{f}_q^2,{f}_b^2,{e}_p^2,\hat{e}_q^2,\hat{e}_b^2)\in \mathcal{D}_d$. Substituting (\ref{eq:Dir-prim-dual}) into (\ref{eq:bild}) yields
\begin{equation}\label{eq-12}
\begin{split}
\!\!\!\!\langle (-1)^{pq+1} e_p^1 \wedge \left(\mathbf{d}_\mathrm{i} \hat{e}_q^2  + \mathbf{d}_\mathrm{b} \hat{e}_b^2 \right ) + \hat{e}_q^1 \wedge \mathbf{d} e_p^2 + (-1)^{pq+1}e_p^2\wedge \left(\mathbf{d}_\mathrm{i} \hat{e}_q^1 +\mathbf{d}_\mathrm{b} \hat{e}_b^1\right) +\hat{e}_q^2 \wedge \mathbf{d} e_p^1, K \rangle\\
+  (-1)^{p} \langle \hat{e}_b^1 \wedge e_p^2 + \hat{e}_b^2 \wedge e_p^1, \partial K \rangle\,.
\end{split}
\end{equation}
By the anticommutativity of the primal-dual wedge product on $K$
\begin{equation*}
\begin{split}
\langle \hat{e}_q^1 \wedge \mathbf{d} e_p^2, K \rangle&= (-1)^{q(p-1)} \langle \mathbf{d} e_p^2\wedge \hat{e}_q^1 ,K\rangle\\
\langle \hat{e}_q^2 \wedge \mathbf{d} e_p^1, K \rangle&= (-1)^{q(p-1)} \langle \mathbf{d} e_p^1\wedge \hat{e}_q^2 ,K\rangle\,,
\end{split}
\end{equation*}
and on the boundary $\partial K$
\begin{equation*}
\begin{split}
\langle \hat{e}_b^1 \wedge e_p^2, \partial K \rangle&= (-1)^{(p-1)(q-1)} \langle e_p^2 \wedge \hat{e}_b^1 ,\partial K\rangle\\
\langle \hat{e}_b^2 \wedge e_p^1, \partial K \rangle&= (-1)^{(p-1)(q-1)} \langle e_p^1 \wedge \hat{e}_b^2 ,\partial K\rangle\,,
\end{split}
\end{equation*}
the expression (\ref{eq-12}) can be rewritten as
\begin{equation*}\label{eq-12n}
\begin{split}
(-1)^{q(p-1)}\langle \mathbf{d}e_p^2 \wedge \hat{e}_q^1 + (-1)^{n-p}e_p^2 \wedge \left(\mathbf{d}_\mathrm{i} \hat{e}_q^1  + \mathbf{d}_\mathrm{b} \hat{e}_b^1 \right ) , K \rangle\\
+(-1)^{q(p-1)}\langle \mathbf{d}e_p^1 \wedge \hat{e}_q^2 + (-1)^{n-p}e_p^1 \wedge \left(\mathbf{d}_\mathrm{i} \hat{e}_q^2  + \mathbf{d}_\mathrm{b} \hat{e}_b^2 \right ) , K \rangle\\
+(-1)^{p+(p-1)(q-1)}\langle \hat{e}_b^1 \wedge e_p^2+ \hat{e}_b^2 \wedge e_p^1,\partial K\rangle\,.
\end{split}
\end{equation*}
According to the discrete summation by parts formula (\ref{eq:sumbypart}), the following holds
\begin{equation*}
\begin{split}
&\!\!\langle \mathbf{d}e_p^2 \wedge \hat{e}_q^1 \!+\! (-1)^{n-p}\!e_p^2 \wedge \!\left(\mathbf{d}_\mathrm{i} \hat{e}_q^1  + \mathbf{d}_\mathrm{b} \hat{e}_b^1 \right ) , K \rangle\!=\!\langle e_p^2\wedge \hat{e}_b^1, \partial K\rangle\\
&\!\!\langle \mathbf{d}e_p^1 \wedge \hat{e}_q^2 \!+ \!(-1)^{n-p}e_p^1 \wedge \!\left(\mathbf{d}_\mathrm{i} \hat{e}_q^2  + \mathbf{d}_\mathrm{b} \hat{e}_b^2 \right ) , K \rangle\!=\!\langle e_p^1\wedge \hat{e}_b^2, \partial K\rangle\,.
\end{split}
\end{equation*}
Hence, (\ref{eq-12}) is equal to $0$, and thus $\mathcal{D}_d \subset \mathcal{D}_d^\perp$.

Since $\mathrm{dim}\,\mathcal{F}_{p,q}^d=\mathrm{dim}\,\mathcal{E}_{p,q}^d=\mathrm{dim}\,\mathcal{D}_d$, and $\langle\!\langle, \rangle\!\rangle_{d}$ is a non-degenerate form, $\mathcal{D}_d =\mathcal{D}_d^\perp$. 
\end{proof}

\begin{rmk}As with the continuous setting, the simplicial Dirac structure is algebraically compositional. Since the simplicial Dirac structure $\mathcal{D}_d $ is a finite-dimensional constant Dirac structure, it is integrable.
\end{rmk}

The other possible discrete analogue of the Stokes-Dirac structure is defined on the spaces
\begin{equation*}
\begin{split}
\tilde{\mathcal{F}}_{p,q}^d&=\Omega_d^p (K)\times \Omega_d^q( \star_\mathrm{i} K)\times \Omega_d^{n-p}(\partial (\star K))\\
\tilde{\mathcal{E}}_{p,q}^d&=\Omega_d^{n-p}( \star_\mathrm{i} K)\times \Omega_d^{n-q}( K)\times \Omega_d^{n-q}(\partial K)\,.
\end{split}
\end{equation*}
A natural discrete mirror of the bilinear form (\ref{eq-aj9}) in this case is a symmetric pairing on the product space $\tilde{\mathcal{F}}_{p,q}^d\times \tilde{\mathcal{E}}_{p,q}^d$ defined by
\begin{equation*}
\begin{split}
\langle\!\langle (&\underbrace{{f}_p^1,\hat{f}_q^1,{\hat{f}}_b^1}_{\in \tilde{\mathcal{F}}_{p,q}^d},\underbrace{{\hat{e}}_p^1,{e}_q^1,{e}_b^1}_{\in \tilde{\mathcal{E}}_{p,q}^d}), ({f}_p^2,\hat{f}_q^2,\hat{f}_b^2,\hat{e}_p^2,{e}_q^2,{e}_b^2)\rangle\!\rangle_{\tilde{d}} \\
&= \langle \hat{e}_p^1\wedge {f}_p^2+ {e}_q^1\wedge {\hat{f}}_q^2+ {\hat{e}}_p^2\wedge {f}_p^1+ {e}_q^2\wedge \hat{f}_q^1,K \rangle+\langle {e}_b^1\wedge \hat{f}_b^2+ {e}_b^2\wedge \hat{f}_b^1,\partial K \rangle\,.
  \end{split}
\end{equation*}

\begin{thm}\label{th:pdDirac}
The linear space $\tilde{\mathcal{D}}_d$ defined by 
\begin{equation}\label{eq:Dir-prim-dual2}
\begin{split}
&\tilde{\mathcal{D}}_d=\big\{ ({f}_p, \hat{f}_q, {f}_b,{e}_p, {e}_q, {e}_b)\in\tilde{ \mathcal{F}}_{p,q}^d\times \tilde{\mathcal{E}}_{p,q}^d\big |\\
& \left(\!\!\begin{array}{c} {f}_p \\ {f}_q\end{array}\!\!\right)=\left(\!\!\begin{array}{cc}0 & (-1)^{pq+1}  \mathbf{d} \\  \mathbf{d}_\mathrm{i} & 0\end{array}\!\!\right)\left(\!\!\begin{array}{c} \hat{e}_p \\ {e}_q\end{array}\!\!\right)+\left(\!\!\begin{array}{c}  0\\ \mathbf{d}_\mathrm{b} \end{array}\!\!\right) \hat{f}_b\,,\\
&\begin{array}{c}\,\,~e_b ~\end{array}= ~(-1)^{p}{e}_q|_{\partial K}\,
\big \}\,
 \end{split}
\end{equation}
is a Dirac structure with respect to the bilinear pairing $\langle\!\langle, \rangle\!\rangle_{\tilde{d}}$.
\end{thm}

In the following section, the simplicial Dirac structures (\ref{eq:Dir-prim-dual}) and (\ref{eq:Dir-prim-dual2}) will be used as \emph{terminus a quo} for the geometric formulation of spatially discrete port-Hamiltonian systems.

\section{Port-Hamiltonian Dynamics on a Simplicial Complex}\label{Sec5}
In the continuous theory, a distributed-parameter port-Hamiltonian system is defined with respect to the Stokes-Dirac structure (\ref{eq-7Cont}) by imposing constitutive relations. On the other hand, in the discrete framework one can define an open Hamiltonian system with respect to the simplicial Dirac structure $\mathcal{D}_d$ or the simplicial structure $\tilde{\mathcal{D}}_d$. The choice of the structure has immediate consequence on the open dynamics since it restricts the choice of freely chosen boundary efforts or flows. Firstly, we define dynamics with respect to the structure (\ref{eq:Dir-prim-dual}) and (\ref{eq:Dir-prim-dual2}). Then, in the manner of finite-dimensional port-Hamiltonian systems, we include energy dissipation by terminating some of the ports by resistive elements.

\subsection{Port-Hamiltonian systems}
Let a function $\mathcal{H}: \Omega_d^p(\star_\mathrm{i} K)\times \Omega_d^q(K)\rightarrow \mathbb{R}$ stand for the Hamiltonian 
$(\hat{\alpha}_p,\alpha_q)\mapsto \mathcal{H}(\hat{\alpha}_p,\alpha_q)$, with $\hat{\alpha}_p \in \Omega_d^p(\star_\mathrm{i} K)$ and $\alpha_q\in \Omega_d^q(K)$. The value of the Hamiltonian after arbitrary variations of $\hat{\alpha}_p$ and $\alpha_q$ for $\delta\hat{\alpha}_p\in \Omega_d^p(\star_\mathrm{i} K)$ and $\delta{\alpha}_q\in \Omega_d^q( K)$, respectively, can, by Taylor expansion, be expressed as
\begin{equation}
\begin{split}
\mathcal{H}(\hat{\alpha}_p+\delta\hat{\alpha}_p,\alpha_q+\delta\alpha_q)&=\mathcal{H}(\hat{\alpha}_p,\alpha_q)+\langle \frac{\partial \mathcal{H}}{\partial \hat{\alpha}_p} \wedge \delta\hat{\alpha}_p + \hat{\frac{\partial \mathcal{H}}{\partial \alpha_q}}\wedge \delta\alpha_q ,K\rangle\\&~~+\,\textrm{higher~order~terms~in}~\delta\hat{\alpha}_p,\delta\alpha_q\,.
 \end{split}
\end{equation}
Here, it is important to emphasize that the variations $\delta\hat{\alpha}_p,\delta\alpha_q$ are not restricted to vanish on the boundary.

A time derivative of $\mathcal{H}$ along an arbitrary trajectory $t\rightarrow (\hat{\alpha}_p(t),\alpha_q(t))\in  \Omega_d^p(\star_\mathrm{i} K)\times  \Omega_d^q( K)$, $t\in \mathbb{R}$, is
\begin{equation}
\begin{split}
\frac{\mathrm{d}}{\mathrm{d}t}\mathcal{H}(\hat{\alpha}_p,\alpha_q)=\langle \frac{\partial \mathcal{H}}{\partial \hat{\alpha}_p} \wedge \frac{\partial \hat{\alpha}_p}{\partial t} + \hat{\frac{\partial \mathcal{H}}{\partial \alpha_q}}\wedge \frac{\partial \alpha_q }{\partial t},K\rangle\,.
 \end{split}
\end{equation}

The relation between the simplicial Dirac structure (\ref{eq:Dir-prim-dual}) and time derivatives of the variables are
\begin{equation}
\begin{split}
\hat{f}_p=-\frac{\partial \hat{\alpha}_p}{\partial t}\,,~~ f_q=-\frac{\partial \alpha_q }{\partial t}\,,
 \end{split}
\end{equation}
while the coenergy variables are set
\begin{equation}
\begin{split}
e_p=\frac{\partial \mathcal{H}}{\partial \hat{\alpha}_p}\,,~~ \hat{e}_q=\hat{\frac{\partial \mathcal{H} }{\partial \alpha_q}}\,.
 \end{split}
\end{equation}

This allows us to define the spatially discrete, and thus finite-dimensional, port-Hamiltonian system on a simplicial complex $K$ (and its dual $\star K$) by
\begin{equation}\label{eq-38d}
\begin{split}
\!\!\!\!\!\!\!\left(\!\!\!\begin{array}{c}-\frac{\partial\hat{\alpha}_p}{\partial t} \\ -{\frac{\partial\alpha_q}{\partial t}}\end{array}\!\!\!\right)&=\left(\!\!\begin{array}{cc}0 & (-1)^{r}  \mathbf{d}_\mathrm{i}\\  \mathbf{d} & 0\end{array}\!\!\right)\!\!\left(\!\!\begin{array}{c} {\frac{\partial \mathcal{H}}{\partial \hat{\alpha}_p}} \\\hat{\frac{\partial \mathcal{H}}{\partial \alpha_q}}\end{array}\!\!\right)+\!(-1)^{r}\!\!\left(\!\!\begin{array}{c}  \mathbf{d}_\mathrm{b} \\ 0 \end{array}\!\!\right) \hat{e}_b\,,\\
\begin{array}{c}~~\,\,f_b ~~\end{array}&= ~(-1)^{p}\frac{\partial \mathcal{H}}{\partial \hat{\alpha}_p}\bigg|_{\partial K}\,,
\end{split}
\end{equation}
where $r=pq+1$.

It immediately follows that $\frac{\mathrm{d}\mathcal{H}}{\mathrm{d}t}= \langle \hat{e}_b \wedge {f}_b, \partial K \rangle $, enunciating a fundamental property of the system: the increase in the energy on the domain $|K|$ is equal to the power supplied to the system through the boundary $\partial K$ and $\partial (\star K)$. Due to its structural properties, the system (\ref{eq-38d}) can be called a spatially-discrete time-continuous boundary control system with $\hat{e}_b$ being the boundary control input and $f_b$ being the output.

An alternative formulation of a spatially discrete port-Hamiltonian system is given in terms of the simplicial Dirac structure (\ref{eq:Dir-prim-dual2}). We start with the Hamiltonian function $
({\alpha}_p, \hat{\alpha}_q)\mapsto \mathcal{H}({\alpha}_p, \hat{\alpha}_q)$, where ${\alpha}_p \in \Omega_d^p(K)$ and $\hat{\alpha}_q\in \Omega_d^q(\star _\mathrm{i}K)$.
In a similar manner as in deriving (\ref{eq-38d}), we introduce the port-Hamiltonian system
\begin{equation}\label{eq-38dN}
\begin{split}
\!\!\!\!\left(\!\!\begin{array}{c}\!\!-\frac{\partial{\alpha}_p}{\partial t} \\ \!\!-{\frac{\partial \hat{\alpha}_q}{\partial t}}\!\!\!\end{array}\right)&\!=\!\left(\!\!\begin{array}{cc}0 & \!\!\!(-1)^{pq+1}  \mathbf{d}\\  \mathbf{d}_\mathrm{i} & 0\end{array}\!\!\right)\!\left(\!\!\begin{array}{c} \hat{{\frac{\partial \mathcal{H}}{\partial{\alpha}_p}}} \\ {\frac{\partial \mathcal{H}}{\partial \hat{\alpha}_q}}\end{array}\!\!\right)\!+\!\left(\!\!\begin{array}{c} 0\\  \mathbf{d}_\mathrm{b}  \end{array}\!\!\right) \hat{f}_b\,,\\
\begin{array}{c}~~\,\,e_b ~~\end{array}&= ~(-1)^{p}\frac{\partial \mathcal{H}}{\partial \hat{\alpha}_q}|_{\partial K}\,.
\end{split}
\end{equation}

In contrast to (\ref{eq-38d}), in the case of the formulation (\ref{eq-38dN}), the boundary flows $\hat{f}_b$ can be considered to be freely chosen, while the boundary efforts $e_b$ are determined by the dynamics.

\subsection{Modelling dissipation}
Incorporation of dissipation parallels the continuous case and for the present moment we shall consider only dissipation modeled by port termination. As an illustration, consider a mapping $\hat{R}_d: \Omega_d^q(K)\rightarrow \Omega_d^{n-q}(\star_\mathrm{i} K)$ that satisfies
$$
\langle \hat{R}_d(f_q)\wedge f_q,K\rangle \geq 0\,~~ \forall f_q\in \Omega_d^q(K)\,.
$$ 
Furthermore, let $\hat{R}_d=R *$ with $R$ being a positive real constant.

In case of the simplicial Dirac structure (\ref{eq:Dir-prim-dual}), introduce the relation
$$
\hat{e}_q=-(-1)^{q(n-q)}\hat{R}_d(f_q)=-(-1)^{q(n-q)}R * f_q\,,
$$
as well as associate to every primal $(n-p)$-cell an energy storage effort variable and to every dual $p$-cell a sign consistent energy flow leading to
\begin{equation*}
\begin{split}
\hat{f}_p=-\frac{\partial \hat{x}_p}{\partial t}\,,~~ e_p=\frac{\partial \mathcal{H} }{\partial \hat{x}}\,,~~\, \hat{x}\in \Omega_d^p(\star_\mathrm{i} K)\,,
 \end{split}
 \end{equation*}
with $ \mathcal{H} $ being a total stored energy.

This leads to relaxation dynamics of a diffusion process
\begin{equation*}
\begin{split}
\frac{\partial \hat{x}}{\partial t} =(-1)^{q-1} R \,\mathbf{d}_\mathrm{i}* \mathbf{d} \frac{\partial \mathcal{H} }{\partial \hat{x}}+(-1)^{pq} \mathbf{d}_\mathrm{b} \hat{e}_b\,,
 \end{split}
 \end{equation*}
with
\begin{equation*}
\begin{split}
\frac{\mathrm{d}H}{\mathrm{d}t}&= \langle \frac{\partial \mathcal{H}}{\partial \hat{x}}\wedge \frac{\partial \hat{x}}{\partial t} , K \rangle\\
&= \langle \frac{\partial \mathcal{H}}{\partial \hat{x}}\wedge \left( (-1)^{q-1} R\, \mathbf{d}_\mathrm{i} * \mathbf{d} \frac{\partial \mathcal{H}}{\partial \hat{x}} +(-1)^{pq} \mathbf{d}_\mathrm{b} \hat{e}_b\right), K \rangle\\
&=-R\langle \mathbf{d}\frac{\partial \mathcal{H}}{\partial \hat{x}} \wedge *\mathbf{d} \frac{\partial H}{\partial \hat{x}},K\rangle+(-1)^p\langle \hat{e}_\mathrm{b} \wedge \frac{\partial \mathcal{H}}{\partial \hat{x}}, \partial K \rangle\\\
&\leq  \langle \hat{e}_b \wedge {f}_b, \partial K \rangle\,.
 \end{split}
 \end{equation*}

Let $\hat{f}_p=*f_p=-\frac{\partial x}{\partial t}$, $f_p,x\in \Omega_d^0(K)$, that is $p=n$ and $q=1$. As the stored energy take $\mathcal{H}=\frac{1}{2}\langle x\wedge *x, K\rangle$. Then
\begin{equation*}
\begin{split}
\frac{\partial x}{\partial t} &= R *\mathbf{d}_\mathrm{i}* \mathbf{d} x=-R \mbox{\boldmath$\delta$} \mathbf{d} x+(-1)^n\mathbf{d}_\mathrm{b} \hat{e}_b =-R\Delta x+(-1)^n\mathbf{d}_\mathrm{b} \hat{e}_b\,,
 \end{split}
 \end{equation*}
where $\Delta$ is the Laplace operator, $\Delta: \Omega_d^0(K)\rightarrow  \Omega_d^0(K)$. One needs to be careful here with the minus sign since by the chosen convention $\Delta x=-\mathrm{div\,grad}\, x$ \cite{Abraham}.
The boundary flow is $f_b=(-1)^n x|_{\partial K}$.

\begin{rmk}
Consider a diffusion process on a one-dimensional simplicial complex $K$ with a dual $\star K$ that is also a one-dimensional simplicial complex (the spatial domain is identical to the domain of the telegraph equations, confer to Figure~\ref{fig:K1}). This means that $n=q=1$ and $p=0$. The operator $\hat{R}_d$ is a positive definite operator that maps the set of primal edges into the set of dual nods. The resulting dissipative port-Hamiltonian system is
\begin{equation*}
\begin{split}
\frac{\partial \hat{x}}{\partial t}&= \mathbf{d}_\mathrm{i} \hat{R}_d(\mathbf{d} x)
 \end{split}
 \end{equation*}
for $\hat{e}_b=0$, which conduces to the standard compartmental model. This can be extended to structure-preserving discretization of reaction-diffusion systems as hinted in \cite{Seslija}.
\end{rmk}

\section{Matrix Representations for Linear Port-Hamiltonian Systems on a Simplicial Complex and Error Analysis}\label{Sec:matrixrep}
Discrete exterior calculus can be implemented using the formalism of linear algebra. All discrete $k$-forms can be stored into a vector with entries assuming the values that those forms take on the ordered set of $k$-simplices. The boundary operator is a linear mapping from the space of $k$-simplices to the space of $(k-1)$-simplices and can be represented by a sparse matrix containing only $\pm1$ elements, while the exterior derivative is its transpose. There is a number of different Hodge star implementations, but the so-called mass-lumped is the simplest, with the Hodge star being a diagonal matrix.

\subsection{Matrix representations of linear operators}
Any discrete differential form $\alpha^k\in \Omega_d^k(K)$ is uniquely characterized by its coefficient vector $\vec{\alpha}\in \Lambda^k$, where $\Lambda^k=\mathbb{R}^{N_k}$, $N_k=\mathrm{dim}\Omega_d^k(K)$ is the number of $k$-simplices. Similarly, for a $\hat{\beta}_\mathrm{i}\in \Omega_d^{n-k}(\star_\mathrm{i}K)$ the vector representation is $\vec{\beta}_i \in \Lambda_{n-k}=\mathbb{R}^{N_k}$. Representing discrete forms by their coefficient vectors induces a matrix representation for linear operators (see e.g. \cite{Desbrun,Hiptmair}).

The exterior derivative $\mathrm{d}: \Omega_d^k(K)\rightarrow \Omega_d^{k+1}(K)$ is represented by a matrix $D^k\in \mathbb{R}^{N_{k+1}\times N_{k}}$, which is the transpose of the incidence matrix of $k$-faces and $(k+1)$-faces of the primal mesh \cite{Desbrun,Desbrun2}. The discrete derivative $\mathbf{d}: \Omega_d^k(\star K)\rightarrow \Omega_d^{k+1}(\star_\mathrm{i} K)$ in the matrix notation is the transpose of the incidence matrix of the dual mesh denoted by $\hat{D}\in{\mathbb{R}}^{N_{k+1}\times N_{k}}$,which can be, as we shall soon show, decomposed as $\hat{D}=( D_{i} \,\vdots\, D_{b} )^\textsc{t}$ with $D_i$ and $D_b$ being matrix representations of $\mathrm{d}_i$ and $\mathrm{d}_b$, respectively.

The exterior product  $\wedge: \Omega_d^k(K)\times \Omega_d^{n-k}(\star_\mathrm{i} K) \rightarrow \Omega_d^{n}(V_k(K))$ between $\alpha\in \Omega_d^k(K)$ and $\hat{\beta}\in \Omega_d^{n-k}(\star_\mathrm{i} K)$ can be written as
\begin{equation*}
\begin{split}
\langle \alpha \wedge \hat{\beta}, K \rangle &=\vec{\alpha}^\textsc{t} W_k^{n-k} \vec{\beta} = (-1)^{k(n-k)}\vec{\beta}^\textsc{t} \widehat{W}_{n-k}^k \vec{\alpha}\\
&= (-1)^{k(n-k)}\langle \hat{\beta} \wedge \alpha, K \rangle\,,
\end{split}
\end{equation*}
where $W_k^{n-k}, \widehat{W}_{n-k}^k\in \mathbb{R}^{N_k\times N_k}$ and $W_k^{n-k}=(-1)^{k(n-k)}(\widehat{W}_{n-k}^k)^\textsc{t}$.

A crucial ingredient for supplying the result of Theorem~\ref{th:pdDirac} is a discrete summation by parts formula (\ref{eq:sumbypart}), which in the context of the simplicial Dirac structures can be rewritten as
\begin{equation*}\label{eq:sumbypartsF}
\begin{split}
\langle \mathrm{d} e_p \wedge \hat{e}_q, K\rangle +(-1)^{k-1} \langle e_p\wedge (\mathrm{d}_i \hat{e}_q + \mathrm{d}_b \hat{e}_b ), K \rangle = \langle e_p \wedge \hat{e}_b, \partial K \rangle\,,
\end{split}
\end{equation*}
where $e_p\in \Omega_d^{k-1}(K)$, $\hat{e}_q\in \Omega_d^{n-k}(\star_\mathrm{i} K)$, $\hat{e}_b\in \Omega_d^{n-k}(\star_\mathrm{b} K)$ for $q=k$ and $p=n-k+1$.

Representing the discrete forms by the corresponding coefficient vectors $\vec{e}_p\in\Lambda^{k}$, $\vec{e}_q\in \Lambda_{n-k}$,  $\vec{e}_b\in\Lambda_{b,{n-k}}$, $\vec{f}_p\in\Lambda_{n-k+1} $, $\vec{f}_q\in \Lambda^{k}$,  $\vec{f}_b\in\Lambda_{b}^{k-1}$ gives rise to the matrix representation of (\ref{eq:sumbypartsF})
\begin{equation*}
\begin{split}
\left( D^{k-1} \vec{e}_p\right)^\textsc{t} W_k^{n-k} \vec{e}_q + (-1)^{k-1} \vec{e}_p^\textsc{t} W_{k-1}^{n-k+1}\left( D_i^{n-k} \vec{e}_q + D_b^{n-k} \vec{e}_b\right) \\= \left(T^{k-1}\vec{e}_p\right)^{\textsc{t}} W_{b, k-1}^{n-k} \vec{e}_b\,.
\end{split}
\end{equation*}
Here the matrix $T^{k-1}\in \mathbb{R}^{\mathrm{dim}\, \Lambda_b^{k-1}\times N_{k-1}}$ is a trace operator of $(k-1)$-forms on the boundary of the primal simplicial complex $K$.

After regrouping we have
\begin{equation*}
\begin{split}
 \vec{e}_p^\textsc{t} \left( \left(D^{k-1}\right)^\textsc{t}  W_k^{n-k}   + (-1)^{k-1} W_{k-1}^{n-k+1} D_i^{n-k} \right)\vec{e}_q + (-1)^{k-1}  \vec{e}_p^\textsc{t} W_{k-1}^{n-k+1} D_b^{n-k} \vec{e}_b \\
 = \left(T^{k-1}\vec{e}_p\right)^{\textsc{t}} W_{b, k-1}^{n-k} \vec{e}_b\,.
\end{split}
\end{equation*}

Choosing $W_k^{n-k}=I_{N_k}$, $W_{k-1}^{n-k+1}=I_{N_{k-1}}$ and $W_{b, k-1}^{n-k}=I_{\mathrm{dim}\, \Lambda_b^{k-1}}$ implies the well-known relation \cite{Desbrun2}
\begin{equation*}
\begin{split}
	D_i^{n-k}&=(-1)^k \left( D^{k-1} \right)^\textsc{t}\,,
\end{split}
\end{equation*}
while interestingly enough the dual boundary operator is a dual of the primal trace operator
\begin{equation*}
\begin{split}
	D_b^{n-k}&=(-1)^{k-1}\left( T^{k-1}\right)^\textsc{t}\,.
\end{split}
\end{equation*}

The discrete Hodge operator $* : \Omega_d^k(K)\rightarrow \Omega_d^{n-k}(\star_i K)$ has the following matrix representation \cite{Desbrun,Desbrun2}
\begin{equation*}
M^k \vec{\alpha} = \vec{\beta} ~~\mathrm{with}~M^k\in \mathbb{R}^{N_k\times N_k} ~~\mathrm{for} ~~\vec{\alpha}\in \Lambda^k\,,\, \vec{\beta}\in \Lambda_{n-k}\,,
\end{equation*}
while the discrete Hodge star from $\Lambda_{n-k}$ to $\Lambda^k$ can be described by
\begin{equation*}
\widehat{M}^{n-k} \vec{\beta} =\widehat{M}^{n-k} M^k \vec{\alpha}\,,
\end{equation*}
where $\widehat{M}^{n-k} M^k = (-1)^{k(n-k)}I_{N_k}$. As in the continuous theory, the discrete Hodge operators are invertible. The norm of $\vec{\alpha}\in\Lambda^k$ induced by the discrete Hodge star is
\begin{equation*}
\| \vec{\alpha} \|_{\Lambda^k}^2=\vec{\alpha}^\textsc{t}M^k \vec{\alpha}=\| (M^k)^{\frac{1}{2}}\vec{\alpha}\|\,.
\end{equation*}

\subsection{Representation of simplicial Dirac structures}
Consider the simplicial Dirac structure (\ref{eq:Dir-prim-dual}). The effort and flow space are
\begin{equation*}
\begin{split}
\mathcal{F}_{p,q}^d\Lambda&=\Lambda_p \times \Lambda^q \times \Lambda_b^{n-p}\,,\\
\mathcal{E}_{p,q}^d\Lambda&=\Lambda^{n-p} \times \Lambda_{n-q} \times \Lambda_{b,{n-q}}\,,
\end{split}
\end{equation*}
and the bilinear form (\ref{eq:bild}) on the space $\mathcal{F}_{p,q}^d\Lambda\times \mathcal{E}_{p,q}^d\Lambda$ is
\begin{equation}\label{eq:bildmatrix}
\begin{split}
\langle\!\langle (&  \vec{f}_p^1,\vec{f}_q^1,\vec{f}_b^1 , \vec{e}_p^1,\vec{e}_q^1,\vec{e}_b^1), (\vec{f}_p^2,\vec{f}_q^2,\vec{f}_b^2,\vec{e}_p^2,\vec{e}_q^2,\vec{e}_b^2)\rangle\!\rangle_d\\
&=   \left(\vec{e}_p^1\right)^\textsc{t} {W}_{n-p}^p \vec{f}_p^2 + \left(\vec{e}_q^1\right)^\textsc{t} \widehat{W}_{n-q}^{q} \vec{f}_q^2+ \left(\vec{e}_p^2\right)^\textsc{t} W_{n-p}^p \vec{f}_p^1+\left(\vec{e}_q^2\right)^\textsc{t} \widehat{W}_{n-q}^q \vec{f}_q^1 \\
&~~~+  \left(\vec{e}_b^1\right)^\textsc{t} W_{b,{n-q}}^{{n-p}} \vec{f}_b^2+ \left(\vec{e}_b^2\right)^\textsc{t} W_{b,{n-q}}^{n-p} \vec{f}_b^1\\
&=\left(\vec{e}_p^1\right)^\textsc{t}  \vec{f}_p^2 + \left(\vec{e}_p^2\right)^\textsc{t}  \vec{f}_p^1+ (-1)^{q(n-q)}\left( \left(\vec{e}_q^1\right)^\textsc{t}  \vec{f}_q^2 +  \left(\vec{e}_q^2\right)^\textsc{t}  \vec{f}_q^1  \right)\\
&~~~+ (-1)^{(n-p)(n-q)} \left( \left(\vec{e}_b^1\right)^\textsc{t} \vec{f}_b^2+  \left(\vec{e}_b^2\right)^\textsc{t} \vec{f}_b^1  \right)\,,
  \end{split}
\end{equation}
where we took $W_{n-p}^{p}=I_{N_p}$, $\widehat{W}_{n-q}^q=(-1)^{q(n-q)}I_{N_q}$, $\widehat{W}_{b,{n-q}}^{n-p}=(-1)^{(n-p)(n-q)}I_{N^b_{n-p}}$.

The matrix representation of the simplicial Dirac structure (\ref{eq:Dir-prim-dual}) is
\begin{equation}\label{eq:DiracMatrixRep}
\begin{split}
\left(\begin{array}{c}\vec{f}_p \\ \vec{f}_q \\\vec{f}_b\end{array}\right)&= 
\left(\begin{array}{ccc}
0 & (-1)^{pq+1}D_i^{n-q} & (-1)^{pq+1}D_b^{n-q} \\
D^{n-p} & 0 & 0 \\
(-1)^p T^{n-p} & 0 & 0\end{array}\right)
\left(\begin{array}{c}\vec{e}_p \\ \vec{e}_q \\ \vec{e}_b\end{array}\right)\\
&=\left(\begin{array}{ccc}
0 & (-1)^{q(p+1)+1}D^{n-p} & (-1)^{q(n-1)}(T^{n-p})^\textsc{t} \\
D^{n-p} & 0 & 0 \\
(-1)^p T^{n-p} & 0 & 0\end{array}\right)
\left(\begin{array}{c}\vec{e}_p \\ \vec{e}_q \\ \vec{e}_b\end{array}\right)\,.
\end{split}
\end{equation}

\subsection{Error analysis}\label{Sec:convergence}
In this section we consider the spatial discretization of a linear distributed-parameter port-Hamiltonian system of the form
\begin{equation}\label{sec:lin-inf}
\begin{split}
-*_p\frac{\partial {{ e}}_p^{\mathrm{c}}}{\partial t}&=(-1)^{pq+1} \mathrm{d}e_q^{\mathrm{c}}\,\\
-*_q\frac{\partial {e}_q^{\mathrm{c}}}{\partial t}&=\mathrm{d}e_p^{\mathrm{c}}\\
e_q^{\mathrm{c}}|_{\partial|K|}&=e_b^{\mathrm{c}}\\
f_b^{\mathrm{c}}&=(-1)^p e_p^{\mathrm{c}}|_{\partial |K|}
\end{split}
\end{equation}
on an $n$-dimensional polytope $|K|$. The operators $*_p$ and $*_q$ are the Hodge stars spawned by Riemannian metrics. 

Note that all continuous (spatially undiscretized) quantities are labeled by a superscript c, for example, $e_p^{\mathrm{c}}$ and $e_q^{\mathrm{c}}$ are the continuous efforts. The approach to convergence analysis we take here is that of \cite{Hiptmair1}.

The discrete analogue of (\ref{sec:lin-inf}) defined with respect to the simplicial Dirac structure (\ref{eq:Dir-prim-dual}) is
\begin{equation}\label{sec:findimsystem}
\begin{split}
-{\widehat{M}}_p\dot{\vec{e}}_p&=(-1)^{pq+1}\left( D_i^{n-q} \vec{e}_q + D_b^{n-q} \vec{e}_b \right)\,\\
-{M}_q  \dot{\vec{e}}_q &=D^{n-p} \vec{e}_p\\
\vec{f}_b&=(-1)^p T^{n-p}\vec{e}_p\,,
\end{split}
\end{equation}
where $\widehat{M}_p\in \mathbb{R}^{N_p\times N_p}$ and $M_q\in \mathbb{R}^{N_q\times N_q}$ are diagonal Hodge matrices. A dot over a variable denotes the time derivative.

Integrate the first equation over dual $p$-cells and the second over primal $q$-faces to obtain
\begin{equation}\label{sec:findimsystem2}
\begin{split}
-\widetilde{M}_p\dot{\vec{e}}^*_p-\dot{\vec{R}}_p&=(-1)^{pq+1}\left( D_i^{n-q} \vec{e}^*_q + D_b^{n-q} \vec{e}^*_b \right)\,\\
-{M}_q \dot{\vec{e}}^*_q-\dot{\vec{R}}_q&=D^{n-p} \vec{e}^*_p\\
\vec{f}^*_b&=(-1)^p T^{n-p}\vec{e}^*_p\,,
\end{split}
\end{equation}
where $\vec{e}_p^*$ and $\vec{e}_q^*$ are integral forms on the primal mesh and its circumcentric dual, while $\dot{\vec{R}}_p$ and $\dot{\vec{R}}_q$ are time derivatives of the residues of the Hodge operator approximations given by
\begin{equation*}
\begin{split}
\int_{\hat{\sigma}_k^p} *_p \dot{e}_p^{\mathrm{c}}&= \widehat{M}_{p,k}\dot{\vec{e}}_{p,k}+\dot{\vec{R}}_{p,k}\\
\int_{{\sigma}_l^q} *_q \dot{e}_q^{\mathrm{c}}&= {M}_{q,l}\dot{\vec{e}}_{q,l}+\dot{\vec{R}}_{q,l}\,,
\end{split}
\end{equation*}
with subscripts $k$ and $l$ acting as selectors for vector components.

Define discrete energy errors as $\delta \vec{e}_p=\vec{e}_p^*-\vec{e}_p$ and $\delta \vec{e}_q=\vec{e}_q^*-\vec{e}_q$, and the output error as $\delta \vec{f}_b=\vec{f}_b^*-\vec{f}_b$.

Subtracting (\ref{sec:findimsystem}) from (\ref{sec:findimsystem2}) leads to
\begin{equation}\label{sec:erroran3}
\begin{split}
-{\widehat{M}}_p \delta\dot{\vec{e}}_p-\dot{\vec{R}}_p&=(-1)^{pq+1}\left( D_i^{n-q} \delta \vec{e}_q + D_b^{n-q} \delta\vec{e}^*_b \right)\,\\
-{M}_q \delta \dot{\vec{e}}_q-\dot{\vec{R}}_q&=D^{n-p} \delta \vec{e}_p\\
\delta\vec{f}_b&=(-1)^p T^{n-p}\delta \vec{e}_p\,,
\end{split}
\end{equation}
since $\delta \vec{e}_b= \vec{e}_b^*-\vec{e}_b=(\int_{\star \sigma_b^{n-q}} e_b^{\mathrm{c}}-\langle e_b, \star \sigma_b^{n-q}\rangle)_{\star \sigma_b^{n-q}\in\partial(\star K)}=0$.

Multiplying the first equation in (\ref{sec:erroran3}) by $\delta \vec{e}_p$ and the second by $\delta \vec{e}_q$ gives
\begin{equation*}
\begin{split}
-\langle \delta \vec{e}_p, {\widehat{M}}_p \delta \dot{\vec{e}}_p\rangle -\langle \delta \vec{e}_p,\dot{\vec{R}}_p\rangle&=(-1)^{pq+1}\langle  \delta \vec{e}_p, D_i^{n-q} \delta \vec{e}_q\rangle\,\\
-\langle \delta \vec{e}_q, {M}_q \delta \dot{\vec{e}}_q\rangle -\langle \delta \vec{e}_q,\dot{\vec{R}}_q \rangle &=\langle \delta \vec{e}_q, D^{n-p} \delta \vec{e}_p\rangle\,.
\end{split}
\end{equation*}
Then we have
\begin{equation*}
\begin{split}
-\langle\delta \vec{e}_p, {\widehat{M}}_p \delta \dot{\vec{e}}_p\rangle -\langle  \delta \vec{e}_p, \dot{\vec{R}}_p\rangle
-\langle \delta \vec{e}_q, {M}_q \delta \dot{\vec{e}}_q\rangle -\langle\delta \vec{e}_q, \dot{\vec{R}}_q\rangle \\
=(-1)^{pq+1}\langle  \delta \vec{e}_p, D_i^{n-q} \delta \vec{e}_q\rangle +\langle \delta \vec{e}_q, D^{n-p} \delta \vec{e}_p\rangle  =0\,.
\end{split}
\end{equation*}
That is
\begin{equation}\label{sec:lin-residues}
\langle \delta \vec{e}_p, {\widehat{M}}_p \delta \dot{\vec{e}}_p\rangle
+\langle \delta \vec{e}_q, {M}_q \delta \dot{\vec{e}}_q\rangle  =-\langle \delta \vec{e}_p,\dot{\vec{R}}_p\rangle-\langle\delta \vec{e}_q \dot{\vec{R}}_q\rangle\,.
\end{equation}

Integration of (\ref{sec:lin-residues}) from $0$ to $t_f$ yields
\begin{equation*}
\begin{split}
\frac{1}{2}\| {\widehat{M}}_p^{\frac{1}{2}} \delta{\vec{e}}_p({t_f})\|^2 +\frac{1}{2} \| {M}_q^{\frac{1}{2}} \delta \vec{e}_q({t_f})\|^2  &=-\int_0^{{t_f}}\langle \delta \vec{e}_p(\tau) , \dot{\vec{R}}_p(\tau)\rangle+\langle\delta \vec{e}_q(\tau), \dot{\vec{R}}_q(\tau)\rangle d\tau\\
&\leq \int_0^{t_f} \| \dot{\vec{R}}_p(\tau)\| \|\delta \vec{e}_p(\tau) \|+ \| \dot{\vec{R}}_q(\tau)\| \|\delta \vec{e}_q(\tau)\| d\tau\,.
\end{split}
\end{equation*}

Let $t^*$ be such that
\begin{equation*}
\begin{split}
\| {\widehat{M}}_p^{\frac{1}{2}} \delta{\vec{e}}_p(t^*)\|^2 + \| {M}_q^{\frac{1}{2}} \delta \vec{e}_q(t^*)\|^2 =\max_{0\leq t\leq t_f} \| \widehat{M}^{\frac{1}{2}}_p \| \,\|\delta \vec{e}_p (t)\| + || M_q^{\frac{1}{2}} \| \| \delta \vec{e}_q \|\,,
\end{split}
\end{equation*}
then
\begin{equation*}
\begin{split}
&\left( \| \widehat{M}_p^{\frac{1}{2}} \delta{\vec{e}}_p(t^*)\| + \| {M}_q^{\frac{1}{2}} \delta \vec{e}_q(t^*)\| \right )^2 \leq 2 \left( \| {\widehat{M}}^{\frac{1}{2}}_p \delta \vec{e}_p (t^*)\|^2 + \| M_q^{\frac{1}{2}} \delta \vec{e}_q(t^*) \|^2 \right) \\
&~~\leq 4 \int_0^{t_f} \| \dot{\vec{R}}_p(\tau)\| \|\delta \vec{e}_p(\tau) \|+ \| \dot{\vec{R}}_q(\tau)\| \|\delta \vec{e}_q(\tau)\| d\tau\\
&~~\leq 4 \int_0^{t_f} \left(  \| {\widehat{M}}^{\frac{1}{2}}_p \delta \vec{e}_p (t)\|+  \| {M}^{\frac{1}{2}}_q \delta \vec{e}_q (t)\|\right) \left( \| {\widehat{M}}^{-\frac{1}{2}}_p \dot{\vec{R}}_p (\tau)\|+  \| {M}^{-\frac{1}{2}}_q  \dot{\vec{R}}_q(\tau)\| \right) d \tau\;.
\end{split}
\end{equation*}
It follows that
\begin{equation*}
\begin{split}
 \| \widehat{M}_p^{\frac{1}{2}} \delta{\vec{e}}_p(t^*)\| + \| {M}_q^{\frac{1}{2}} \delta \vec{e}_q(t^*)\| &\leq 4 \int_0^{t_f} \| {\widehat{M}}^{-\frac{1}{2}}_p \dot{\vec{R}}_p (\tau)\|+  \| {M}^{-\frac{1}{2}}_q  \dot{\vec{R}}_q(\tau)\|   d \tau\,.
\end{split}
\end{equation*}
Thus
\begin{equation*}
\begin{split}
 \| \delta{\vec{e}}_p(t^*)\| + \|   \delta \vec{e}_q(t^*)\| &\leq 4 \left( \| \widehat{M}_p^{-\frac{1}{2}}\|_\infty + \| M_q^{-\frac{1}{2}} \|_\infty \right) \int_0^{t_f} \| {\widehat{M}}^{-\frac{1}{2}}_p \dot{\vec{R}}_p (\tau)\|+  \| {M}^{-\frac{1}{2}}_q  \dot{\vec{R}}_q(\tau)\|   d \tau\,.
\end{split}
\end{equation*}

Estimation of the residues $\vec{R}_p$ and $\vec{R}_q$ can be conducted by employing Bramble-Hilbert techniques in the case of a weak formulation, or using a Taylor's expansion of the efforts under the standard smoothness assumptions \cite{Hiptmair}. For the results on the estimates of the Hodge star in one, two and three dimension the reader is invited to consult \cite{Hiptmair} and references therein.

\section{Physical Examples}\label{Sec7}
In this section we formulate discrete analogues of distributed-parameter port-Hamiltonian systems on a three-, two-, and one-dimensional manifold.
\subsection{Maxwell's equations}
Let $K$ be a well-centered $3$-dimensional manifold-like simplicial complex with circumcentric dual $\star K$, endowed with a discrete Riemannian metric. Mirroring the continuous case \cite{vdSM02}, we formulate the discrete Maxwell's equations in terms of discrete differential forms, and then we demonstrate that the underpinning differential/gauge structure is preserved.

The energy variables are chosen such that they live on the discrete manifolds that are dual to one another. For instance, we choose the magnetic (field) induction $2$-form to be defined on the primal simplicial complex $K$ as $\alpha_q=B\in \Omega_d^2(K)$ and the electric induction $2$-form $\hat{\alpha}_p=\hat{D}\in \Omega_d^2(\star_\mathrm{i} K)$. This means that $B$ and $\hat{D}$ do not reside at the same discrete locations, but rather at separate faces of staggered lattices.

\begin{rmk}In the case of a spatio-temporal discretization based on the asynchronous variational integrator scheme, as proposed in \cite{Stern}, the electric and magnetic induction are also defined at different time locations leading to improved numeric performance (for more details refer to \cite{Stern}).
\end{rmk}

The coenergy variables are chosen coherently as implied by the choice of the energy variables such that the discrete Maxwell's equations fit the simplicial Dirac structure (\ref{eq:Dir-prim-dual}) for $n=3$, $p=q=2$. This entails that the magnetic field intensity $\hat{e}_q=\hat{H}\in \Omega_d^1(\star_\mathrm{i} K)$ and the electric intensity $e_p=E\in \Omega_d^1(K)$, as such, are related to the energy variables via
\begin{equation*}
\begin{split}
\hat{D}&=* \epsilon E\\
B&=* \mu \hat{H}\,,
\end{split}
\end{equation*}
where $\epsilon$ and $\mu$ denote the constant electric and magnetic permittivity, respectively.

The corresponding simplicial Dirac structure is 
\begin{equation}\label{eq:Dir-prim-dualME}
\begin{split}
\left(\begin{array}{c}\hat{f}_p \\ {f}_q\end{array}\right)&=\left(\begin{array}{cc}0 & -\mathbf{d}_\mathrm{i}\\  \mathbf{d} & 0\end{array}\right)\left(\begin{array}{c} {e}_p \\\hat{e}_q\end{array}\right)-\left(\begin{array}{c} \mathbf{d}_\mathrm{b} \\ 0\end{array}\right) \hat{e}_b\,\\
\begin{array}{c}~~\,\,f_b ~~\end{array}&= ~ \,e_p|_{\partial K}\,.
 \end{split}
\end{equation}
The Hamiltonian is $
\mathcal{H}=\frac{1}{2} \langle  E\wedge \hat{D}+ \hat{H}\wedge B,K\rangle\,,
$ or expressed only in terms of the primal forms as $\mathcal{H}=\frac{1}{2}\langle E \wedge * \epsilon E + \frac{1}{\mu} * B\wedge B, K \rangle$.

Under the assumption that there is no current in the medium, the spatially discretized Maxwell's equations with respect to the simplicial Dirac structure (\ref{eq:Dir-prim-dualME}) in the port-Hamiltonian form are given by
\begin{equation}\label{eq:Dir-prim-dualME1}
\begin{split}
\!\!\left(\!\!\begin{array}{c}- \frac{\partial \hat{D}}{\partial t} \\ -\frac{\partial B}{\partial t}\end{array}\right)&=\left(\begin{array}{cc}0 & \!-\mathbf{d}_\mathrm{i}\\  \mathbf{d} & 0\end{array}\!\!\right)\!\left(\!\begin{array}{c} \frac{\partial \mathcal{H}}{\partial \hat{D}} \\ \hat{\frac{\partial \mathcal{H}}{\partial B }}\end{array}\!\right)-\left(\begin{array}{c} \mathbf{d}_\mathrm{b} \\ 0\end{array}\right) \hat{e}_b\,\\
\begin{array}{c}~~\,\,f_b ~~\end{array}&= ~ \,\frac{\partial \mathcal{H}}{\partial \hat{D}}\bigg|_{\partial K}\,.
 \end{split}
\end{equation}
The readily proved energy balance is $\frac{\mathrm{d}\mathcal{H}}{\mathrm{d}t}= \langle \hat{e}_b \wedge{f}_b, \partial K \rangle $. Incorporating a nonzero current density into the discrete Maxwell's equations is straightforward as in the continuous case.

\subsection{Two-dimensional wave equation}
In order to demonstrate practically that we do not face a problem of interconnection of the elementary Dirac structures encountered in the mixed finite element method, as reported by \cite{Voss} (see pages 183--196), we consider the simplicial Dirac structure behind the discretized two-dimensional wave equation. The normalized wave equation is given by
$$\frac{\partial^2 \phi}{\partial t^2} - \Delta \phi=0\,,
$$
where $\phi$ is a smooth $0$-form on a compact surface $M\subset \mathbb{R}^2$ with a closed boundary, and $\Delta$ is the Laplace operator. This equation, together with nonzero energy flow, can be formulated as a port-Hamiltonian system with boundary port variables \cite{Talasila,Golo}.

The energy variables of the discretized system are chosen as follows: the kinetic momentum is a dual $2$-form whose time derivative is set to be $\hat{f}_p$, the elastic strain is a primal $1$-form with time derivative corresponding to $f_q$, the coenergy variables are a primal $0$-form $e_p$ and a dual $1$-form $\hat{e}_q$. Such a formulation of the discrete wave equation is consonant with the simplicial Dirac structure (\ref{eq:Dir-prim-dual}) for the case when $p=n=2$ and $q=1$. We shall, nevertheless, practically confirm the arguments of Theorem~\ref{th:pdDirac} in a simple low-dimensional model.

Consider a ring of counterclockwise oriented triangles that could be, say, obtained by a very coarse discretization of a disk. The dual of the central vertex $v_0$ is its Voronoi region, while the duals of the boundary vertices are the convex boundary pentagons. The orientation of the primal edges is chosen as indicated in Figure~\ref{fig:pentagon}. The orientation of the dual edges is induced such that the basis of the primal and dual cells combined give the orientation of the embedding space that, in our case, has been given by the right-hand rule (for more on orientation see pages 11--22 of \cite{Hirani}).
\begin{figure}
\centering
\vspace{-0.4cm}
      \hspace{0cm}\includegraphics[width=8.5cm]{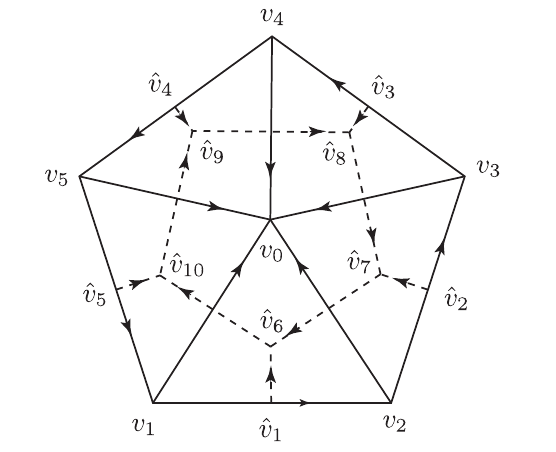}
      \vspace{-0.2cm}
  \caption{A simplicial complex $K$ consists of five triangles arranged into a pentagon. The dual edges introduced by subdivision are shown dotted.}\label{fig:pentagon}
\end{figure}

It suffices to check the power conserving property of the founding Dirac structure. We need to show that
\begin{equation*}
\begin{split}
\langle e_p \wedge \hat{f}_p + \hat{e}_q \wedge f_q, K \rangle +\langle \hat{e}_b \wedge {f}_b, \partial K \rangle=0 \,.
 \end{split}
\end{equation*}
This is equivalent to the validity of the following relation
\begin{equation*}
\begin{split}
\langle \mathbf{d} e_p \wedge \hat{e}_q + e_p \wedge (\mathbf{d}_\mathrm{i} \hat{e}_q + \mathbf{d}_\mathrm{b} \hat{e}_b), K \rangle = \langle e_p \wedge \hat{e}_b, \partial K \rangle \,.
 \end{split}
\end{equation*}

We calculate
\begin{equation*}\label{eq-example1}
\begin{split}
\langle \mathbf{d} e_p \wedge \hat{e}_ q ,K \rangle &= \sum_{\sigma^1\in K}\langle \mathbf{d} e_p   ,\sigma^1\rangle \langle \hat{e}_q, \star \sigma^1\rangle= \sum_{\sigma^1\in K}\langle  e_p   ,\partial \sigma^1\rangle \langle \hat{e}_q, \star \sigma^1\rangle=\sum_{{\substack{\sigma^1\in K\\ \sigma^0\prec \sigma^1}}}\langle  e_p   ,\sigma^1\rangle \langle \hat{e}_q, \star \sigma^1\rangle\\
&= \left(  e_p(v_2)-e_p(v_1)\right) \hat{e}_q([\hat{v}_1,\hat{v}_6])+\left(  e_p(v_3)-e_p(v_2)\right) \hat{e}_q([\hat{v}_2,\hat{v}_7])\\
&~~~+\left(  e_p(v_4) - e_p(v_3)\right) \hat{e}_q([\hat{v}_3,\hat{v}_8])+\left(  e_p(v_5)-e_p(v_4)\right) \hat{e}_q([\hat{v}_4,\hat{v}_9])\\
&~~~+\left(  e_p(v_1) - e_p(v_5)\right) \hat{e}_q([\hat{v}_5,\hat{v}_{10}])+\left(  e_p(v_0)-e_p(v_1)\right) \hat{e}_q([\hat{v}_6,\hat{v}_{10}])\\
&~~~+\left(  e_p(v_0) - e_p(v_2)\right) \hat{e}_q([\hat{v}_7,\hat{v}_6])+\left(  e_p(v_0)-e_p(v_3)\right) \hat{e}_q([\hat{v}_8,\hat{v}_7])\\
&~~~+\left(  e_p(v_0) - e_p(v_4)\right) \hat{e}_q([\hat{v}_9,\hat{v}_8])+\left(  e_p(v_0)-e_p(v_5)\right) \hat{e}_q([\hat{v}_{10},\hat{v}_9])\,
 \end{split}
\end{equation*}
and
\begin{equation*}\label{eq-example2}
\begin{split}
\langle  e_p\wedge (\mathbf{d}_\mathrm{i}\hat{e}_q+\mathbf{d}_\mathrm{b}\hat{e}_b), K \rangle & =\sum_{\star \sigma^0\in \star K} \langle e_p,\sigma^0\rangle \langle \mathbf{d}_\mathrm{i} \hat{e}_q+\mathbf{d}_\mathrm{b} \hat{e}_b, \star \sigma^0\rangle\\
&=\sum_{\substack{\star \sigma^0\in \star K}} \langle e_p, \sigma^0\rangle\left( \langle \hat{e}_q, \partial_\mathrm{i} (\star \sigma^0)\rangle + \langle \hat{e}_b, \partial_\mathrm{b} (\star \sigma^0)\rangle\right)\\
&= e_p(v_1) \Big ( \hat{e}_q([\hat{v}_1,\hat{v}_6])+ \hat{e}_q([\hat{v}_6,\hat{v}_{10}])- \hat{e}_q([\hat{v}_5,\hat{v}_{10}])+\hat{e}_b([\hat{v}_5,\hat{v}_1]) \Big)\\
&~~+ e_p(v_2) \Big ( \hat{e}_q([\hat{v}_2,\hat{v}_7])+ \hat{e}_q([\hat{v}_7,\hat{v}_{6}])- \hat{e}_q([\hat{v}_1,\hat{v}_{6}])+\hat{e}_b([\hat{v}_1,\hat{v}_2]) \Big)\\
&~~+ e_p(v_3) \Big( \hat{e}_q([\hat{v}_3,\hat{v}_8])+ \hat{e}_q([\hat{v}_8,\hat{v}_{7}])- \hat{e}_q([\hat{v}_2,\hat{v}_{7}])+\hat{e}_b([\hat{v}_2,\hat{v}_3]) \Big)\\
&~~+ e_p(v_4) \Big( \hat{e}_q([\hat{v}_4,\hat{v}_9])+ \hat{e}_q([\hat{v}_9,\hat{v}_{8}])- \hat{e}_q([\hat{v}_3,\hat{v}_{8}])+\hat{e}_b([\hat{v}_3,\hat{v}_4]) \Big)\\
&~~+ e_p(v_5) \Big( \hat{e}_q([\hat{v}_5,\hat{v}_{10}])+ \hat{e}_q([\hat{v}_{10},\hat{v}_{9}])- \hat{e}_q([\hat{v}_4,\hat{v}_{9}])+\hat{e}_b([\hat{v}_4,\hat{v}_5]) \Big)\\
&~~+ e_p(v_0) \Big( -\hat{e}_q([\hat{v}_7,\hat{v}_6])- \hat{e}_q([\hat{v}_8,\hat{v}_{7}])- \hat{e}_q([\hat{v}_9,\hat{v}_{8}])-\hat{e}_q([\hat{v}_{10},\hat{v}_9])\\
&~~~~~~~~~~~~~~~ -\hat{e}_q([\hat{v}_{6},\hat{v}_{10}]) \Big)\,.
 \end{split}
\end{equation*}

After summation of the last two relations, all terms, except those associated with the primal and dual boundary, cancel out, leading to
\begin{equation}
\begin{split}
\langle \mathbf{d} e_p \wedge \hat{e}_ q ,K \rangle + \langle  e_p \wedge ( \mathbf{d}_\mathrm{i} \hat{e}_ q +\mathbf{d}_b \hat{e}_ b ),K \rangle &= e_p(v_1)\hat{e}_b([\hat{v}_5,\hat{v}_1]) + e_p(v_2)\hat{e}_b([\hat{v}_1,\hat{v}_2])\\  
&~~~+ e_p(v_3)\hat{e}_b([\hat{v}_2,\hat{v}_3]) +e_p(v_4)\hat{e}_b([\hat{v}_3,\hat{v}_4])  \\  
&~~~+ e_p(v_5)\hat{e}_b([\hat{v}_4,\hat{v}_5])\,.
 \end{split}
\end{equation}
This confirms that the boundary terms genuinely live on the boundary of $|K|$.

\subsection{Telegraph equations}
We consider an ideal lossless transmission line on a $1$-dimensional simplicial complex. The energy variables are the charge density ${q}\in \Omega_d^1(K)$, and the flux density $\hat{\phi}\in\Omega_d^1(\star  K)$, hence $p=q=1$. The Hamiltonian representing the total energy stored in the transmission line with discrete distributed capacitance $C$ and discrete distributed inductance $L$ is
\begin{equation*}
\begin{split}
\mathcal{H}=  \langle \frac{1}{2C} {q} \wedge * {q}  +  \frac{1}{2L}  \hat{\phi} \wedge * \hat{\phi}  ,K\rangle \,,
 \end{split}
\end{equation*}
with co-energy variables: $\hat{e}_p=\hat{\frac{\partial \mathcal{H}}{\partial{q}}}=*\frac{{q}}{C}=\hat{V}$ representing voltages and ${e}_q=\frac{\partial \mathcal{H}}{\partial \hat{\phi }}=*\frac{\hat{\phi}}{L}=I$ currents.

Selecting ${f}_p=-\frac{\partial{q}}{\partial  t}$ and $\hat{f}_q=-\frac{\partial \hat{\phi}}{\partial  t}$ leads to the port-Hamiltonian formulation of the telegraph equations
\begin{equation*}\label{eq:Dir-prim-dualMET}
\begin{split}
\left(\begin{array}{c} - \frac{\partial {q} }{\partial t} \\ -\frac{\partial \hat{\phi} }{\partial t}\end{array}\right)&=\left(\begin{array}{cc} 0 & \mathbf{d} \\  \mathbf{d}_\mathrm{i} & 0\end{array}\right) \left(\begin{array}{c} *\frac{q}{C} \\  *\frac{ \hat{\phi} }{L}  \end{array}\right)+\left(\begin{array}{c} 0 \\  \mathbf{d}_\mathrm{b}  \end{array}\right) \hat{f}_b\,\\
\begin{array}{c}~~\,\,e_b ~~\end{array} &= ~ - * \frac{\hat{\phi}}{L} \bigg|_{\partial K}\,.
 \end{split}
\end{equation*}
In the case we wanted to have the electrical current as the input, the charge and the flux density would be defined on the dual mesh and the primal mesh, respectively. Instead of the port-Hamiltonian system in the form (\ref{eq-38d}), that is (\ref{eq:Dir-prim-dualMET}), the discretized telegraph equations would be in the form (\ref{eq-38d}). The free boundary variable is always defined on the boundary of the dual cell complex. 

Note that the structure (\ref{eq:Dir-prim-dualMET}) is in fact a Poisson structure on the state space $\Omega_d^1( K)\times \Omega_d^1( \star K)$. This will become obvious when we present this structure in a matrix representation. Before that, it is illustrative to demonstrate how the pairings between primal and dual forms can be rather easily calculated.

\begin{figure}
\centering
    \includegraphics[width=9cm]{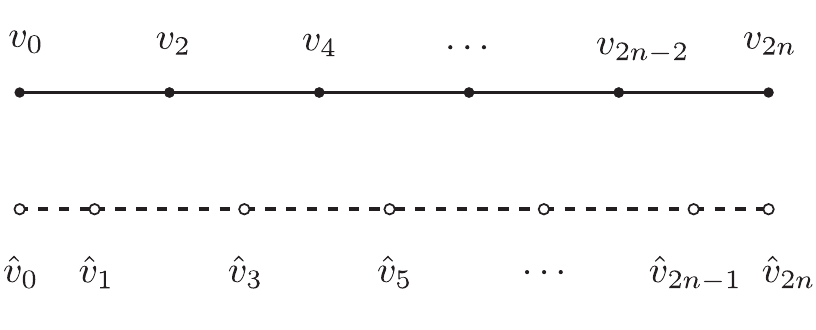}\\ \vspace{-0.3cm}
  \caption{The primal $1$-dimensional simplicial complex $K$ with even nodes indices and its dual $\star K$ with odd indices, both with conventional orientation of one simplices (from the node with a lower-index to the higher-index node). By construction, the nodes $\hat{v}_0$ and $\hat{v}_{2n}$ are added to the boundary as previously explained to insure that the boundary of the dual is the dual of the boundary, i.e., $\partial (\star K)=\star (\partial K)$. }{\label{fig:K1}}
\end{figure}

Using the notation from Figure~\ref{fig:K1}, we have
\begin{equation*}
\begin{split}
\langle \mathbf{d} e_q\wedge \hat{e}_p, K \rangle & =\sum_{\sigma^1\in K} \langle \mathbf{d}e_q,\sigma^1\rangle \langle \hat{e}_p, \star \sigma^1\rangle=\sum_{\substack{\sigma^1\in K}} \langle e_q,\partial \sigma^1\rangle \langle \hat{e}_p, \star \sigma^1\rangle\\
&=[e_q(v_2)-e_q(v_0)]\hat{e}_p(\hat{v}_1)+[e_q(v_4)-e_q(v_2)]\hat{e}_p(\hat{v}_3)+\ldots\\
&\,~~~+[e_q(v_{2n-2})-e_q(v_{2n-4})]\hat{e}_p(\hat{v}_{2n-3})\\
&\,~~~+[e_q(v_{2n})-e_q(v_{2n-2})]\hat{e}_p(\hat{v}_{2n-1})\\
&=-e_q(v_0)\hat{e}_p(\hat{v}_1)-e_q(v_2) [\hat{e}_p(\hat{v}_3)-\hat{e}_p(\hat{v}_1)]\\
&\,~~~-e_q(v_4) [\hat{e}_p(\hat{v}_5)-\hat{e}_p(\hat{v}_3)]~~~-\ldots \\
&\,~~~- e_q(v_{2n-2}) [\hat{e}_p(\hat{v}_{2n-1})-\hat{e}_p(\hat{v}_{2n-3})]+e_q(v_{2n}) \hat{e}_p(\hat{v}_{2n-1})\,
 \end{split}
\end{equation*}
and
\begin{equation*}
\begin{split}
\langle  e_q\wedge (\mathbf{d}_\mathrm{i}\hat{e}_p+\mathbf{d}_\mathrm{b}\hat{f}_b), K \rangle 
& =\sum_{\star \sigma^0\in \star K} \langle e_q,\sigma^0\rangle \langle \mathbf{d}_\mathrm{i} \hat{e}_p+\mathbf{d}_\mathrm{b} \hat{f}_b, \star \sigma^0\rangle\\
&=\sum_{\substack{\star \sigma^0\in \star K}} \langle e_q, \sigma^0\rangle\left( \langle \hat{e}_p, \partial_\mathrm{i} (\star \sigma^0)\rangle + \langle \hat{f}_b, \partial_\mathrm{b} (\star \sigma^0)\rangle\right)\\
&=e_q({v}_0)[\hat{e}_p(\hat{v}_1)-\hat{f}_b(\hat{v}_0)]+e_q({v}_2) [\hat{e}_p(\hat{v}_3)-\hat{e}_p(\hat{v}_1)]\\
&\,~~~+e_q(v_{4}) [\hat{e}_p(\hat{v}_{5})-\hat{e}_p(\hat{v}_{3})]
+\ldots\\
&\,~~~+ e_q(v_{2n-2}) [\hat{e}_p(\hat{v}_{2n-1})-\hat{e}_p(\hat{v}_{2n-3})] \\
&\,~~~+{e}_q(v_{2n}) [\hat{f}_b(\hat{v}_{2n})-\hat{e}_p(v_{2n-1})]\,.
 \end{split}
\end{equation*}

The arguments of the discrete Stokes theorem (\ref{eq:sumbypart}) are trivially verified
\begin{equation*}
\begin{split}
\langle \mathbf{d} e_q\wedge \hat{e}_p, K \rangle +\langle  e_q\wedge (\mathbf{d}_\mathrm{i}\hat{e}_p+\mathbf{d}_\mathrm{b}\hat{f}_b), K \rangle &=\langle e_q\wedge \hat{f}_b, \partial K \rangle\\&= -e_q(v_0)\hat{f}_b(\hat{v}_0)+e_q(v_{2n})\hat{f}_b(\hat{v}_{2n})\,
  \end{split}
\end{equation*}
showing the power-preserving property of the simplicial Dirac structure (\ref{eq:Dir-prim-dualMET}) which implies that for any $(\hat{e}_p,{f}_p,{e}_q,\hat{f}_q,{e}_b,{\hat{f}}_b)$ in the simplicial structure (\ref{eq:Dir-prim-dualMET}) the following holds
\begin{equation*}
\begin{split}
\langle  \hat{e}_p\wedge  {f}_p, K \rangle +\langle   {e}_q\wedge \hat{f}_q, K \rangle+  \langle  {e}_b \wedge {\hat{f}}_b, \partial K \rangle=0\,.
  \end{split}
\end{equation*}
The energy balance for the transmission line thus is
\begin{equation}
\begin{split}
\frac{\mathrm{d}\mathcal{H}}{\mathrm{d}t}= \langle  {e}_b \wedge \hat{f}_b, \partial K \rangle=  e_b(v_{2n}){\hat{f}}_b(\hat{v}_{2n})-e_b(v_0) \hat{f}_b(\hat{v}_{0}) \,,
 \end{split}
\end{equation}
which demonstrates that the boundary objects genuinely live on the boundary $\partial K$.

\vspace{0.15cm}
\noindent{\textbf{Matrix representation.}} A differential form $e_q\in \Omega_d^0(K)$ is uniquely characterized by its coefficient vector $\vec{e}_q\in{\mathbb{R}}^{n+1}$ since $\mathrm{dim}\, \Omega_d^0(K)=n+1$, similarly $\vec{e}_p,\vec{f}_p\in \mathbb{R}^n$, $\vec{f}_q\in \mathbb{R}^{n+1}$, $\vec{e}_b,\vec{f}_b\in \mathbb{R}^2$. The exterior derivative $\mathbf{d}: \Omega_d^0(K)\rightarrow \Omega_d^1(K)$ is represented by a matrix $D\in \mathbb{R}^{n\times(n+1)}$, which is the transpose of the incidence matrix of the primal mesh \cite{Desbrun,Desbrun2}. The discrete derivative $\mathbf{d}_\mathrm{i}: \Omega_d^0(\star_i K)\rightarrow \Omega_d^1(\star K)$ in the matrix notation is $D_{i}=-D^\textsc{t}$, and $\mathbf{d}_\mathrm{b}: \Omega_d^0(\star_b K)\rightarrow \Omega_d^1(\star K)$ is represented by $D_{b}$, which is the transpose of the trace operator. For the simplicial complex in Figure~\ref{fig:K1}, we have
\begin{equation}
\begin{split}
D&=\left(\begin{array}{rrrrrr}-1 & 1 & 0 & \cdots & 0 & 0 \\0 & -1 & 1 & \cdots & 0 & 0 \\ \, & \, & \, & \ddots & \, & \, \\0 & 0 & 0 & \cdots & -1 & 1\end{array}\right)\,,\\
D_b^\textsc{t}&=\left(\begin{array}{rrrrrr}-1 & 0 & 0 & \cdots & 0 & 0 \\0 & 0 & 0 & \cdots & 0 & 1\end{array}\right)\,.
\end{split}
\end{equation}

Implementing the primal-dual wedge product as a scalar multiplication of the coefficient vectors and by taking for convenience $\vec{e}_b=\left(e_b(v_0),-e_b(v_{2n})\right)^\textsc{t}$, the simplicial Dirac structure of (\ref{eq:Dir-prim-dualMET}) can be represented by
\begin{equation}
\begin{split}
\left(\begin{array}{c}\vec{f}_p \\ \vec{f}_q \\\vec{e}_b\end{array}\right)=\left(\begin{array}{ccc}0 & D & 0 \\-D^\textsc{t} & 0 & D_{b} \\0 & -D_{b}^\textsc{t} & 0\end{array}\right)\left(\begin{array}{c}\vec{e}_p \\ \vec{e}_q \\\vec{f}_b\end{array}\right)\,.
\end{split}
\end{equation}

Expressing flows in terms of efforts and choosing $L=C=1$ for convenience, the matrix formulation of (\ref{eq:Dir-prim-dualMET}) is
\begin{equation}\label{eq:Dir-prim-dualMETmatrix}
\begin{split}
-M_p \dot{\vec{e}}_p&= D \vec{e}_q\\
-\widehat{M}_q \dot{\vec{e}}_q&= -D_i^\textsc{t}\vec{e}_p + D_b \vec{f}_b\\
\vec{e}_b&= -D_b^\textsc{t} \vec{e}_q\,,
 \end{split}
\end{equation}
where $\vec{e}_p\in \Lambda_0$, $\vec{e}_q\in \Lambda^1$, $M_p =\mathrm{diag}\!\left( {h}_1, {h}_3,\dots,{h}_{2n-1} \right)\in \mathbb{R}^{n\times n}$ and ${\widehat{M}}_q  =\mathrm{diag}\!\left( \hat{h}_0, \hat{h}_2,\dots,\hat{h}_{2n} \right)\in \mathbb{R}^{(n+1)\times (n+1)}$, with $h_1=|[v_0,v_2]|$, $h_3=|[v_2,v_4]|$, \dots, $h_{2n-1}=|[v_{2n-2},v_{2n}]|$ and $\hat{h}_0=|[\hat{v}_0,\hat{v}_1]|$, $\hat{h}_2=|[\hat{v}_1,\hat{v}_3]|$, \dots, $\hat{h}_{2n}=|[\hat{v}_{2n-1},\hat{v}_{2n}]|$.

\vspace{0.15cm}
\noindent{\textbf{Error analysis.}} The time derivatives of the $\vec{R}_p$ components are
\begin{equation*}
\begin{split}
\dot{\vec{R}}_{p,l} &= h_{2l-1} \dot{e}_p(\hat{v}_{2l-1}) - \int_{[v_{2l-1},v_{2l}]} * \dot{e}_p^{\mathrm{c}}=h_{2l-1} \dot{e}_p(\hat{v}_{2l-1})-\int_{v_{2l-1}}^{v_{2l}}\dot{e}_p^{\mathrm{c}}(z)\mathrm{d}z
\end{split}
\end{equation*}
The Taylor's expansion of $\dot{e}_p^{\mathrm{c}}$ around $\hat{v}_{2l-1}$ is
\begin{equation*}
\begin{split}
\dot{e}_p^{\mathrm{c}} (z)=\dot{e}_p^{\mathrm{c}} (\hat{v}_{2l-1}) + \frac{\partial \dot{e}_p^{\mathrm{c}}}{\partial z}(\hat{v}_{2l-1})(z-\hat{v}_{2l-1})+\frac{\partial^2 \dot{e}_p^{\mathrm{c}}}{\partial z^2}(\hat{v}_{2l-1})\frac{(z-\hat{v}_{2l-1})^2}{2}+O(z^3)\,
\end{split}
\end{equation*}
Thus
\begin{equation*}
\begin{split}
\int_{v_{2l-2}}^{v_{2l}}\dot{e}_p^{\mathrm{c}} (z) \mathrm{d}z=\dot{e}_p^{\mathrm{c}} (\hat{v}_{2l-1}) h_{2l-1} + \frac{1}{3}\frac{\partial^2 \dot{e}_p^{\mathrm{c}}}{\partial z^2}(\hat{v}_{2l-1})\left(\frac{\hat{h}_{2l}}{2}\right)^3+O(h_{2l}^4)\,,
\end{split}
\end{equation*}
and
\begin{equation*}
\begin{split}
\dot{\vec{R}}_{p,l}&= \frac{1}{3}\frac{\partial^2 \dot{e}_p^{\mathrm{c}}}{\partial z^2}(\hat{v}_{2l-1})\left(\frac{h_{2l}}{2}\right)^3+O(h_{2l}^4)\,.
\end{split}
\end{equation*}

Likewise,
\begin{equation*}
\begin{split}
\dot{\vec{R}}_{q,k+1}&=\hat{h}_{2k} \dot{e}_q({v}_{2k})-\int_{[\hat{v}_{2k-1},\hat{v}_{2k+1}]}* \dot{e}_q^{\mathrm{c}}=\hat{h}_{2k} \dot{e}_q({v}_{2k})-\int_{\hat{v}_{2k-1}}^{\hat{v}_{2k+1}}\dot{e}_q^{\mathrm{c}}(z)\mathrm{d}z\,,
\end{split}
\end{equation*}
where $k=0,1,\dots, n$ and $\hat{v}_{-1}=\hat{v}_0$. In the similar fashion, we obtain
\begin{equation*}
\begin{split}
\dot{\vec{R}}_{q,k+1}&= \frac{1}{3}\frac{\partial^2 \dot{e}_q^{\mathrm{c}}}{\partial z^2}({v}_{2k})\left(\frac{\hat{h}_{2k}}{2}\right)^3+O(\hat{h}_{2k}^4)~\,\mathrm{for}\,~k=1,\dots,n-1\\
\dot{\vec{R}}_{q,1}&= \frac{1}{2}\frac{\partial \dot{e}_q^{\mathrm{c}}}{\partial z}({v}_{0}) \hat{h}_{0}^2 + O(\hat{h}_{0}^3)\\
\dot{\vec{R}}_{q,n+1}&= -\frac{1}{2}\frac{\partial \dot{e}_q^{\mathrm{c}}}{\partial z}({v}_{2n})\hat{h}_{2n}^2 + O(\hat{h}_{2n}^3)\,.
\end{split}
\end{equation*}

It follows that there exist $K_1, K_2\in \mathbb{R}$ such that
\begin{equation}\label{errorteleqs}
\begin{split}
 \| \delta{\vec{e}}_p(t^*)\| + \|   \delta \vec{e}_q(t^*)\| &\leq h K_2 \int_0^{t_f} \| T\frac{\partial \dot{{e}}_q^{\mathrm{c}}}{\partial z} (\tau)\| \mathrm{d} \tau\\
 &~~~+h^2 K_1\int_0^{t_f} \left(\| \frac{\partial^2 \dot{{e}}_p^{\mathrm{c}}}{\partial z^2} (\tau)\|+\| \frac{\partial^2 \dot{{e}}_q^{\mathrm{c}}}{\partial z^2} (\tau)\|\right ) \mathrm{d} \tau +O(h^3)\,
\end{split}
\end{equation}
where $T$ is a trace operator, $T=D_b^\textsc{t}$, and $h=\min \{h_1, h_3,\dots, h_{2n-1}, \hat{h}_0, \hat{h}_2,\dots, \hat{h}_{2n}\}$.

\vspace{0.15cm}
\noindent{\textbf{Numerical example.}} To evince how exactly the discrete model relates to the continuous one, we take the example from Section 5 in \cite{Golo}. 

The spatial domain of the transmission system is the line segment $M=[0,\mathrm{e}-1]$. The distributed capacitance and the distributive inductance are $z\mapsto C^{\mathrm{c}}(z)=\frac{1}{1+z}$ and $z\mapsto L^{\mathrm{c}}(z)=\frac{1}{1+z}$, $z\in M$. On the left-hand side a causal input voltage $t\mapsto u(t)$ is assigned, and at the other end the transmission line is terminated by a load of unit resistance, meaning $*q(t,\mathrm{e}-1)=*\phi(t,\mathrm{e}-1)$. Initial conditions are assumed to be zero, i.e. $q^{\mathrm{c}}(0,z)=0$ and $\phi^{\mathrm{c}}(0,z)=0$ for $z\in Z$. The exact solution for the voltage distribution is $(t,z)\mapsto e_p^{\mathrm{c}}(t,z)=u(t-\mathrm{ln}(z+1))$, for $t\geq0$.

\begin{figure}
\centering
      \includegraphics[width=12.5cm]{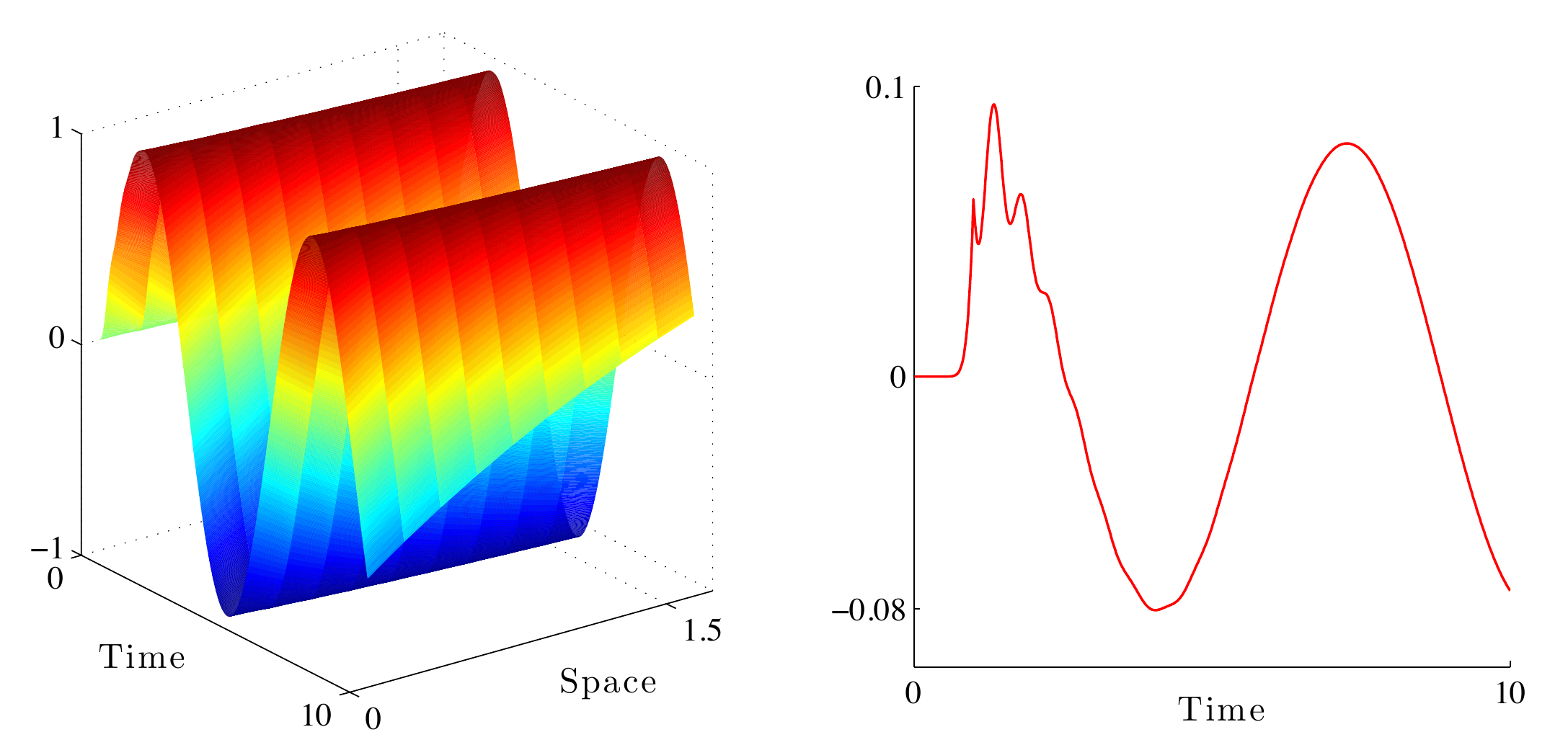}
  \caption{On the left, the voltage distribution $e_p$ for $n=10$, and on the right, the voltage error at the point $z=\mathrm{e}-1$, that is $e_b(v_{20})(t)-\sin(t-1)$ for $t\geq 0$.}\label{fig:NumRes}
\end{figure}

Using equidistant division of $M$ and diagonal Hodge operators, the results of numerical simulation when the input $e_p^{\mathrm{c}}(0,t)=u(t)=\sin t$, $t\geq 0$, are given in Figure~\ref{fig:NumRes}. The time integration technique is Runge-Kutta 4 and the integration step is $0.01$.

All numerical experiments indicate that the discrepancy between the exact value of the voltage and the value obtained by numerical simulation is the greatest at the spatial point $z=\mathrm{e}-1$. Thence, in the left-hand side of Figure~\ref{fig:NumRes} we show this error as a function of time. In all computational experiments, this error, similar to the results in \cite{Golo}, exhibits an oscillatory behavior with the amplitude not exceeding the maximum displayed in the first period.

Repeating simulation experiments for uniform grids of different densities indicates that the accuracy of the proposed method is $1/n$, what comes as no surprise since we worked with diagonal Hodge operators, which are of first-order accuracy as shown in (\ref{errorteleqs}) for the system (\ref{eq:Dir-prim-dualMETmatrix}).

\section{Concluding Remarks}\label{Sec:conclusion}
In the framework of discrete exterior calculus, we have established the theoretical foundation for formulation of time-continuous spatially-discrete port-Hamiltonian systems. The staple fiber of our approach is the formulation of the simplicial Dirac structures as discrete analogues of the Stokes-Dirac structure. These discrete finite-dimensional Dirac structures are the foundation for the definition of open finite-dimensional systems with Hamiltonian dynamics. Such an approach to discretization transfers the essential topological, geometrical, and physical properties from distributed-parameter systems to their finite-dimensional analogues. By preserving the Hamiltonian structure, this methodology utilizes the analysis and control synthesis for the discretized systems.

A number of interesting topics and open questions still need to be addressed. Here we provide a few miscellaneous reflections and some comments on future work.

\vspace{0.15cm}
\noindent{\textbf{Numerical aspects.}} The discrete exterior calculus employed in this paper is founded on the idea of a simplicial complex and its circumcentric dual. While for some problems Delaunay triangulation is desirable since it reduces the maximum aspect of the mesh, for others the construction of circumcentric duals might be too expensive (see \cite{Shewchuk} and references therein). This motivates the development of a discrete calculus on non-simplicial complex meshes, such as a general CW complex \cite{Hatcher} or a rectangular scheme. Although the latter might be inappropriate for geometrically complex objects, a potential advantage would be its conceptual simplicity since the circumcentric dual is again a rectangular mesh.

A major challenge from the numerical analysis standpoint is to offer a careful study of the convergence properties of discrete exterior calculus. Furthermore, it would be desirable to have higher-order discrete analogues of the smooth geometric operators. This primarily pertains to deriving higher-accuracy Hodge star operators, which would possibly in return make structure-preserving discretization more competitive even in the domains where structure is put aside. A recent article \cite{Arnold} reports some significant initial results regarding stability of finite element exterior calculus. The abstract theory is applied to linear elliptic partial differential equations with intention to capture the key structure of de Rham cohomology and as such mainly pertains to the vanishing boundary constraints. Another related publication \cite{Holst} extends the framework of \cite{Arnold} to approximate domains. In the future, in the context of \cite{Arnold,Holst}, it would be interesting to study structure-preserving discretization of port-Hamiltonian systems in the framework of Hilbert complexes.

\vspace{0.15cm}
\noindent
{\textbf{Open discretized systems.}} The Stokes-Dirac structure has proven to be successful in capturing the essential geometry behind many open systems with Hamiltonian dynamics. The concept of the Stokes-Dirac structure as presented in the introduction in order to accommodate some port-Hamiltonian systems, such as the ideal isentropic fluid, needs to be augmented \cite{vdSM02}. The main idea behind these modifications remains to be based on the Stokes theorem. From a structure-preserving discretization point of view, there appears not to be any impediments; nonetheless, in order to discuss these questions in a systematic manner, a unified theory of open infinite-dimensional Hamiltonian systems is needed. The main novelty in discretizing some so formulated general underlying structures might concern their integrability. The simplicial Dirac structures formulated in this paper are constant Dirac structures and as such they satisfy the usual integrability conditions \cite{Courant, Dalsmo, Dorfman}.

An important application of structure-preserving discretization of port-Hamiltonian systems might be in (optimal) control theory, what also prompts a need for time discretization. For closed Hamiltonian systems, it is well-known that asynchronous variational integrators in general cannot preserve the Hamiltonian exactly; however, these integrators, for small time steps, can preserve a nearby Hamiltonian up to exponentially small errors \cite{Lewetal, Lewetal1,Marsden, Marsden1, Marsden2}. An important issue in this context is to study the effects these integrators have on passivity (and losslessness) of open dynamical systems.

\vspace{0.15cm}
\noindent{\textbf{Covariant formulation.}} It is known that, for instance, Maxwell's equations are also consonant with multisymplectic structure since they can be derived from the Hamiltonian variational principle \cite{Abraham,Marsden}. The multysymplectic structure behind Maxwell's equations, unlike the Stokes-Dirac structure, is defined, not on a spatial manifold $M$, but on a spacetime manifold $X$. Here we need to notice that one could define a Stokes-Dirac type structure on a pseudo-Riemannian, say Lorentzian, manifold. In Lorentzian spacetime, the forms $E$ and $B$ can be combined into a single object, the Faraday $2$-form $F=E\wedge \mathrm{d} t+B$. The form $F$ can also be expressed in terms of the electromagnetic potential $1$-form $A$ as $F=\mathrm{d} A$. The Hodge star of $F$ is a dual $2$-form $G=* F=H\wedge \mathrm{d} t-D$, known as the Maxwell's $2$-form. The charge density $\rho$ and current density $J$ can be combined into the source $3$-form $j=J\wedge \mathrm{d} t-\rho$. A well-known relativistically covariant formulation of Maxwell's equations \cite{Abraham} is: $\mathrm{d} F=0$ and $\mathrm{d} G=j$.

In order to relate this formulation to the port-Hamiltonian framework, define the following Stokes-Dirac structure on a Lorentzian manifold $X$ by
\begin{equation}
\begin{split}
\mathcal{D_L}=\big\{& (f_p,f_q,f_b,e_p,e_q,e_b)\in \\
&\Omega^2(X)\times \Omega^3(X)\times \Omega^2(\partial X)\times \Omega^2(X)\times \Omega^1(X)\times \Omega^1(\partial X) \big |\\
&f_p=-\mathrm{d} e_q\,, f_q=\mathrm{d}e_p\,, e_b=-e_q|_{\partial X}\,,f_b=e_p|_{\partial X}\big \}\,.
 \end{split}
\end{equation}
Similar to the proof of Theorem 1, it is easy to verify that $\mathcal{D_L}=\mathcal{D_L}^\perp$, with respect to a natural bilinear form. The dynamics of Maxwell's equations can now be imposed by setting $f_p=F$, $f_q=j$, $e_p=A$, and $e_q=G$. Furthermore, since $\mathrm{d}^2=0$, it follows that $\mathrm{d}f_p=\mathrm{d}\mathrm{d} e_q=0$.

A natural choice for discretization of the structure $\mathcal{D_L}$, in the context of discrete exterior calculus, would be on a simplicial $4$-complex. This would insure a completely covariant formulation of discrete Maxwell's equations, similar to Regge's formalism for producing simplicial approximations of spacetime in numerical general relativity. The relativistic effects in most engineering applications are however negligible, hence for these purposes, by choosing a time coordinate, we can split the Lorentzian manifold into $3+1$ space, whose discrete analogue is a prismal cell complex. Similar to discretization of multisymplectic structures, this would lead to a certain type of asynchronous variational integrator.

An important and challenging avenue for future work is to make an explicit relation between multisymplectic and Stokes-Dirac structures, and then to compare their discrete analogues.

\end{document}